\documentclass[pra,twocolumn,aps,superscriptaddress,longbibliography]{revtex4-1}
\PassOptionsToPackage{usenames,dvipsnames}{color}
\usepackage{amsmath,amsfonts,amssymb,amsthm,graphicx,bbm,enumerate,times,color,textcomp}

\newtheorem{theorem}{Theorem}

\newtheorem{lemma}[theorem]{Lemma}

\newtheorem{remark}[theorem]{Remark}

\usepackage{comment}

\definecolor{rodrigo}{rgb}{0,.4,1}
\newcommand{\ro}[1]{{\color{black}#1}}
\definecolor{henrik}{rgb}{1,.4,0}
\newcommand{\he}[1]{{\color{black}#1}}

\newcommand{\mc}[1]{\mathcal{#1}}

\newcommand{\mb}[1]{\mathbb{#1}}

\newcommand{\e}{\mathrm{e}}

\newcommand{\tr}{\mathrm{Tr}} 

\newcommand{\id}{\mb{I}}

\newcommand{\one}{\mathbf{1}}

\newcommand{\norm}[1]{\left\Vert #1 \right\Vert}

\newcommand{\ket}[1]{\left.\left|{#1}\right.\right\rangle}

\newcommand{\bra}[1]{\left.\left\langle{#1}\right.\right|}

\newcommand{\ketbra}[2]{\ket{#1} \!\! \bra{#2}}
\newcommand{\proj}[1]{\ketbra{#1}{#1}}

\newcommand{\vac}{\mc V}

\newcommand{\fu}{Dahlem Center for Complex Quantum Systems, Freie Universit{\"a}t Berlin, 14195 Berlin, Germany}

\usepackage{hyperref}
\begin{document}
\title{The third law of thermodynamics as a single inequality}
 \author{Henrik Wilming}
 \author{Rodrigo Gallego}
\affiliation{\fu}
\begin{abstract}
The third law of thermodynamics in the form of the unattainability principle states that exact ground-state cooling requires infinite resources. Here we investigate the amount of {non-equilibrium} resources needed for approximate cooling. We consider as resource any system out of equilibrium, allowing for resources beyond the \emph{i.i.d.} assumption and including the input of work as a particular case. We establish in full generality a sufficient and a necessary condition for cooling and show that for a vast class of non-equilibrium resources these two conditions coincide, providing a single necessary and sufficient criterion. Such conditions are expressed in terms of a single function playing a similar role for the third law to the one of the free energy for the second law. From a technical point of view we provide new results about concavity/convexity of certain Renyi-divergences, which might be of independent interest.
\end{abstract}
\maketitle
\section{Introduction}
Pure quantum states are indispensable resources for any task in quantum information processing. However, the third law of thermodynamics (more precisely, the unattainability principle) states that cooling a system exactly to zero temperature requires an infinite amount of resources, being it in the form of time, space, work or some other resource \cite{Nernst1906,Nernst1906_2,Masanes2017,Scharlau2016,Silva2016}. Similarly, no-go theorems have been put forward for the task of bit erasure --which is closely related to ground-state cooling-- showing that no unitary process on a system and a finite dimensional reservoir can bring the system from a mixed to a pure state \cite{DiFranco2013,Wu2013,Ticozzi2014}. However, these no-go results do not say much about the amount of resources needed for approximate cooling. Indeed, in recent times a sizable number of studies deal with different protocols to cool a small quantum system by unitarily acting on a heat bath and a certain number of systems out of equilibrium to be ``used up''  (known under the name of algorithmic or dynamical cooling) \cite{Schulman1999,Boykin2002,Schulman2005,Allahverdyan2011,Raeisi2015} or studying particular models of refrigerating small quantum systems \cite{Allahverdyan2010,Skrzypczyk2010,Skrzypczyk2011,Levy2012,Mari2012,Silva2015,Hofer2016}, including ones that seem to challenge the unattainability principle in terms of required time \cite{Cleuren2012,Allahverdyan2012,Levy2012a,Entin-Wohlman2014,Cleuren2012a,Kolar2012,Kosloff2013,Freitas2017}.

In this work we will focus on quantifying in full generality the expenditure of arbitrary systems out of equilibrium that are needed for approximate cooling while having access to a heat bath. Our scenario is similar to the one considered in algorithmic cooling, but here we treat the full thermodynamics of the problem by allowing for resources with non-trivial Hamiltonians and accounting for the energy conservation of the total process.
We will do this in the resource theoretic framework of quantum thermodynamics \cite{Janzing00,Horodecki2013,Gour2015,Halpern2014,YungerHalpern2016}, which has proven useful to answer a variety of fundamental questions in quantum thermodynamics, such as establishing an infinite family of second laws \cite{Brandao2015}, providing fundamental bounds to single-shot thermodynamics \cite{Horodecki2013,Brandao2015,Faist15,Halpern2015}, providing definitions of work for quantum systems \cite{Horodecki2013,Aberg2013,Gallego2015}, generalizing fluctuation theorems \cite{Alhambra2016,Aberg2016}, elucidating the thermodynamic meaning of negative entropies \cite{delRio2011} and elucidating the role of quantum coherence in thermodynamics \cite{Aberg2014,Faist14,Lostaglio14,Lostaglio2015,Korzekwa2016}.

Recently, there have also been studies from this point of view on the problem of cooling \cite{Masanes2017,Scharlau2016,Silva2016}, however mostly focusing on providing necessary conditions in terms of resources such as time, space or Hilbert-space dimension.

The task of cooling that we are considering can be phrased as finding a cooling protocol between an arbitrary resource described by the state and Hamiltonian $\rho_R$ and $H_R$, respectively, and a target system described by $\rho_S$ and $H_S$ so that $\rho_S$ approximates the ground-state of $H_S$. We will later assume for simplicity that $\rho_S$ is a thermal state -- in this case the goal is to bring its final temperature $T_S$ to a very low value. We will assume that the density matrix of the resource has full rank because otherwise the problem trivializes, since one can, for example, simply swap with a ground-state \footnote{Similar arguments can be made for other states without full rank. For example two copies of a rank-2 state in a 4-dimensional system can be written as a pure state in tensor product with a full-rank state. Such a resource therefore already contains a pure state which can then be mapped to the ground-state.}. We furthermore assume that the target system is initially in thermal equilibrium with some environment. Then the transition, \emph{i.e.} the cooling protocol, can be performed by using a thermal bath at a fixed inverse temperature $\beta$ and performing a global unitary that commutes with the total Hamiltonian, so that energy conservation is properly accounted for. This kind of transitions have been extensively studied and they can be characterized by families of functions $M^\alpha$, the so-called \emph{monotones}, so that a transition is possible if and only if \cite{Brandao2015,Buscemi2017,Gour16}
\begin{equation}\label{eq:secondlaws}
M^\alpha(\rho_R,H_R) \geq M^\alpha(\rho_S,H_S) \;\;\; \forall \\\ \alpha.
\end{equation}
Hence, the problem at hand is in principle hard to characterize since one needs to verify an infinite number of conditions to conclude that a given transition is possible. The main contribution of the present work is to show that in the limit where $T_S$ is sufficiently close to zero --\emph{i.e.} the regime where the (un)attainability problem is formulated-- the infinite set of monotones appearing in \eqref{eq:second_laws} can be essentially reduced to a single monotone. We call this monotone the \emph{vacancy} and it is defined as
\begin{equation}
\mc{V}_{\beta} (\rho,H):=S(\omega_{\beta}(H)\| \rho),
\end{equation}
where $\omega_{\beta}(H)$ is the Gibbs state of at inverse temperature $\beta$ and $S$ is the relative entropy defined as
\begin{align}
  S(\rho\| \sigma) =\tr(\rho\log \rho) - \tr(\rho \log \sigma),
\end{align}
if $\mathrm{supp}(\rho)\subseteq\mathrm{supp}(\sigma)$ and equal to $+\infty$ otherwise.

We find that \emph{sufficient} and \emph{necessary} conditions for cooling, respectively, are given by
\begin{align}\label{eq:thirdlawgeneralintro}
\mc{V}_{\beta}(\rho_R,H_R) - K(\rho_R,H_R,\rho_S,H_S,\beta) &\geq \mc{V}_{\beta}(\rho_S,H_S),\\
\mc{V}_{\beta}(\rho_R,H_R) &\geq \mc{V}_{\beta}(\rho_S,H_S)\label{eq:intronecessary},
\end{align}
where $K (\rho_R,H_R,\rho_S,H_S,\beta) \rightarrow 0$ as $T_S \rightarrow 0$. Hence in the limit of very low temperature cooling $\mc{V}_{\beta}(\rho_S,H_S)$ is the key quantity that determines the fundamental limitations.
{Importantly, $\mc{V}_\beta(\rho_S,H_S)$ diverges as $T_S\rightarrow 0$. The necessary condition \eqref{eq:intronecessary} therefore shows that an infinite amount of resources (as measured by $\mc V_\beta$) is necessary for exact ground-state cooling.} Furthermore we show that for a vast class of resource systems, for example thermal states of coupled harmonic oscillators, the function  $K(\rho_R,H_R,\rho_S,H_S,\beta)$ vanishes identically. Hence
\begin{equation}\label{eq:thirdlawintro}
\mc{V}_{\beta}(\rho_R,H_R) \geq \mc{V}_{\beta}(\rho_S,H_S)
\end{equation}
becomes both a sufficient and necessary condition. That $\mc{V}_\beta$ plays an important role for the third law had been previously found in the setting of \emph{i.i.d.} resources and qubits as target systems in the seminal work of Ref. \cite{Janzing00}. Here, we extend the significance of the quantity $\mc{V}_{\beta}$ to arbitrary scenarios.

Usually, the unattainability principle is formulated with respect to time, arguing that an infinite amount of time (or infinitely many cycles of a periodically working machine) are needed to cool a system exactly to zero temperature. Our results show, for example, that if the non-equilibrium resources are simply hot thermal systems (as in the example of a thermal machine that operates between two heat baths), the system to be cooled and the cooling machine have to effectively interact with infinitely many such resource systems (or all parts of one infinitely large system). This implies that an infinite amout of time is needed, since each such interaction takes a finite time (see \cite{Masanes2017} for a thorough discussion of this point).

Our findings not only serve to pose limitations to protocols of algorithmic cooling, but also suggest a surprising symmetry between the second and third law of thermodynamics. The second law --in its averaged version or in the version of the Jarzynski equality \cite{Jarzynski97}-- can be expressed in terms of the free-energy difference defined as
\begin{equation}\label{eq:secondlaw}
\Delta F_\beta (\rho,H) = \frac{1}{\beta} S(\rho\| \omega_{\beta}(H)).
\end{equation}
In analogy, we show that the third law can be expressed similarly in terms of $\mc{V}_\beta(\rho,H)$ which simply inverts the arguments of the relative entropy in Eq. \eqref{eq:secondlaw} and drops the pre-factor.
This symmetry between the second and third law is quite surprising and hints at the fact that the second and third law can be related to the errors of first and second kind in hypothesis testing \cite{Tomamichel2016}.
We leave the investigation of this deeper relation between the two for future work.

From a technical point of view, our results rely on certain convexity-properties of the function $\alpha \mapsto S_\alpha(\rho || \sigma)$, where $S_\alpha$ are classical Renyi-divergences \cite{Tomamichel2016}. We believe that these results might be of independent interest.

\section{Set-up and general necessary condition}\label{sec:set-up}
In the following we will use the set-up of \emph{catalytic thermal operations} \cite{Janzing00,Horodecki2013,Brandao2015} applied to the task of cooling. In this set-up we imagine to possess a resource given by the pair of state and Hamiltonian $(\rho_R,H_R)$. We can then use an arbitrary thermal bath at inverse temperature $\beta$, that is, a system in a Gibbs state $\omega_{\beta}(H_B)$ of a Hamiltonian $H_B$, and finally an ancillary system, the so-called catalyst with arbitrary state and Hamiltonian $(\sigma_C,H_C)$ in such a way the latter is returned in the same configuration and uncorrelated from the rest of the systems after implementing the protocol. The target system to be cooled is initially assumed to be in thermal equilibrium with the thermal bath and therefore described by a Gibbs state $(\omega_{\beta}(H_S),H_S)$. The total compound $RSBC$ is transformed by a cooling protocol, which consists simply of a unitary transformation $U$ which commutes with the total Hamiltonian.

\he{More formally, we say that there exists a cooling protocol to $\rho_S$ using the resource $(\rho_R,H_R)$ if there exists a \emph{fixed} catalyst $(\sigma_C,H_C)$ and for \emph{any} $\epsilon>0$ there exists a unitary $U$ and a bath Hamiltonian $H_B$ such that}
\begin{equation}\label{eq:thermal_operation}
 \rho'_{RS} \otimes \sigma^\epsilon_C = \tr_{B}\left(U \rho_R \otimes\omega_{\beta}(H_S) \otimes \omega_\beta(H_B)\otimes \sigma_C \; U^\dagger \right)
\end{equation}
with $\tr_{R}(\rho'_{R S})=\rho_S$  \ro{and $\norm{\sigma_C - \sigma_C^\epsilon}_1\leq  \epsilon$}. The only constraint on the unitary $U$ is that it conserves the global energy, i.e.,
\begin{equation}\label{eq:energyconservation}
[U,H_R+H_S+H_B+H_C]=0.
\end{equation}

Note that this formulation of the cooling process contains as a particular case partial cooling in which we do not start with the target in a Gibbs state. In this case, the initial system of $S$, if it is partially cooled before starting the protocol, can be simply incorporated as a part of the resource $R$.

The problem of finding conditions for the existence of a transitions of the form \eqref{eq:thermal_operation} has been studied in Ref. \cite{Brandao2015} for diagonal states, that is, with $[\rho_R,H_R]=0$ and $[\rho_S,H_S]=0$. Throughout this manuscript we will restrict to such diagonal states, but we emphasize that the the necessary condition \eqref{eq:intronecessary} also holds for non-diagonal states as we will see later.

Under the assumption that $\rho_R$ and $\rho_S$ are diagonal, one can show that cooling to a state $\rho_S$ is possible if and only if \cite{Brandao2015}
\begin{align}
\label{eq:second_laws}
S_\alpha(\rho_R || \omega_\beta(H_R)) \geq S_\alpha(\rho_S || \omega_\beta(H_S))\quad \forall \alpha \geq 0,
\end{align}
where $S_\alpha$ are so-called \emph{Renyi-divergences}. The proof of this statement relies simply on the results of Ref. \cite{Brandao2015} together with the additivity of the Renyi-divergences under tensor-products.

An important tool that appears in Eq. \eqref{eq:second_laws} is the concept of a \emph{monotone} of (catalytic) thermal operations \cite{Horodecki2013,Janzing00}. This is any function $f$ which can only decrease under (catalytic) thermal operations. The functions $S_{\alpha}$ appearing in \eqref{eq:second_laws} are monotones under catalytic thermal operations and more generally under any channel that has the Gibbs state as a fixed point. Importantly, any monotone $f$, possibly different from $S_\alpha$, allows us to construct necessary conditions for a given transition. We will now show that $\mc{V}_\beta$ is also a monotone under catalytic thermal operations and derive the corresponding necessary condition for cooling.
\begin{theorem}[Monotonicity and necessary condition]\label{thm:necessary} The vacancy is an additive monotone under catalytic thermal operations. \ro{This has as an implication that}
for any target $(\rho_S,H_S)$ and resource $(\rho_R,H_R)$ --not necessarily diagonal states--, the condition
  \begin{equation}\label{eq:thirdlawnecessary}
  \mc{V}_{\beta}(\rho_R,H_R)\geq \mc{V}_{\beta}(\rho_S,H_S).
  \end{equation}
  is \emph{necessary} for cooling.
\begin{proof}
\ro{Let us first show that $\mc{V}_{\beta}$ is a monotone under catalytic thermal operations. Let us consider an arbitrary transition from state $\rho$ to state $\rho'$ --both with Hamiltonian H-- by catalytic thermal operations, then we will now show that $\mc{V}_{\beta}(\rho,H)\geq \mc{V}_{\beta}(\rho',H)$.}

\he{First note that the vacancy diverges for a state $\rho$ without full rank, \ro{thus \he{the inequality $\mc{V}_{\beta}(\rho,H)\geq \mc{V}_{\beta}(\rho',H)$} is satisfied trivially for those states.} Let us therefore assume that \ro{$\rho$ is a full rank state.}
As was shown in~\cite{Brandao2015}, for any $0\leq\alpha\leq 2$ the Renyi-divergences
\begin{align}
S_\alpha(\rho\| \omega_\beta(H)) := \frac{1}{\alpha-1}\log\tr\left(\rho^\alpha \omega_\beta(H)^{1-\alpha}\right)
\end{align}
are monotonoic under (catalytic) thermal operations for arbitrary states $\rho$. \ro{That is, we have the necessary condition}
\begin{align}
  \label{eq:necessaryproof}
S_\alpha(\rho || \omega_\beta(H)) \geq S_\alpha(\rho' || \omega_\beta(H))\quad \forall 0\leq \alpha \leq 2,
  \end{align}
  By simple algebra, one can show that
\begin{align}
\left.\frac{\partial }{\partial_\alpha}\right|_{\alpha=0} S_\alpha(\rho \| \omega_\beta(H)) = V_\beta(\rho,H) \geq 0.
\end{align}
Taylor-expanding \eqref{eq:necessaryproof} on both sides for any $\alpha>0$ and dividing by $\alpha$ then yields
\begin{align}
V_\beta(\rho,H) + O(\alpha)\geq V_\beta(\rho',H) + O(\alpha),
\end{align}
where $O(\alpha)$ indicates that it is of first order on $\alpha$. Taking $\alpha$ arbitrarily small then yields $V_\beta(\rho,H) \geq V_\beta(\rho',H)$. This proves monotonicity under catalytic thermal operations. Additivity follows directly from the additivity of the relative entropy under tensor products.}

Once we established that the vacancy is a monotone under catalytic thermal operations, we can now derive the necessary condition for cooling by simply applying this condition to the transition that $SR$ undergo in the cooling process:
\begin{align}
&\mc{V}_{\beta}\left(\rho_R \otimes \omega_{\beta}(H_S), H_R+H_S \right) \nonumber \\
\nonumber &\quad \geq \mc{V}_{\beta}\left(\rho'_{RS} , H_R+H_S \right)\\
\nonumber &\quad \geq \mc{V}_{\beta}\left( \omega_{\beta}(H_R) \otimes \rho_S, H_R+H_S \right).
\end{align}
The last inequality follows from the fact that one can always replace the state on any system by an uncorrelated thermal state at the heat bath's temperature using a thermal operation. Using additivity of the vacancy and the fact that $\mc{V}_{\beta}(\omega_{\beta}(H),H)=0$, we obtain the necessary condition~\eqref{eq:thirdlawnecessary}.
\end{proof}
\end{theorem}
We emphasize that the necessary condition~\eqref{eq:thirdlawnecessary} is derived in full generality and it applies to any full-rank state $\rho_R$ and any state $\rho_S$, possibly not diagonal in the eigenbasis of $H_S$.

The monotone $\mc{V}_{\beta}$ was first introduced in Ref. \cite{Janzing00}. Its relevance for the unattainability principle is clear since if $\rho_S$ does not have full support, then the r.h.s. of \eqref{eq:thirdlawnecessary} diverges. Hence, exact cooling is impossible unless the resource $\rho_R$ does not have full support either. In the particular case of $\rho_R= \bigotimes _{i=1}^n \varrho^i$ where $\varrho^i$ has full support, the condition \eqref{eq:thirdlawnecessary} already tell us that we need infinite resources --infinite n in this case-- for exact cooling. Hence, such a simple analysis already suggests that the quantity $\mc{V}_\beta$ plays a crucial role for the limitations on cooling.

To summarize, we have seen, building upon previous literature, that the vacancy $\mc{V}_{\beta}$ establishes completely general necessary conditions for cooling. However, for necessary and sufficient conditions one should in principle verify an infinite number of inequalities given by \eqref{eq:second_laws}. Our contribution will be to show that these infinite number of inequalities can be reduced to a single one, which can also be expressed in terms of the vacancy for a sufficiently cold final state $\rho_S$. Furthermore, we will show that for a large family of resource systems the single sufficient condition that we find coincides with the necessary condition \eqref{eq:thirdlawnecessary}. Hence, the limits on cooling are entirely ruled by the function $\mc{V}_{\beta}$. This holds for large classes of finite systems, with possibly correlated and interacting subsystems.

\section{General sufficient conditions for cooling}
The process of cooling laid out in the previous section can in principle be applied to any final state $\rho_S$. We will now assume for simplicity that the final state, as it corresponds to a cooling process, is of the form $\rho_S=\omega_{\beta_S}(H_S)$ with $\beta_S$ very large. We can then derive the following completely general \emph{sufficient} condition for cooling.

\begin{theorem}[General sufficient condition for cooling] \label{thm:generaltheorem} For every choice of $\beta$ and $H_S$ there is a critical $\beta_{\mathrm{cr}}>0$ such that for any $\beta_S > \beta_{\rm cr}$ and full-rank resource $(\rho_R,H_R)$ the condition
\begin{align}\label{eq:main_result}
\mc V_\beta(\rho_R,H_R) - K(\beta_S,\beta,\rho_R,H_R,H_S) \geq \mc V_\beta(\omega_{\beta_S}(H_S),H_S)
\end{align}
is \emph{sufficient} for cooling. The positive semidefinite function $K$ has the property $K(\beta_S,\beta,\rho_R,H_R) \rightarrow 0$ as $\beta_S\rightarrow \infty$ for any fixed $\beta,H_R,\rho_R>0$ and $H_S$.
\end{theorem}

The proof of the theorem is given in Sec. \ref{sec:proofmaintheorem}. Nonetheless we provide a sketch of the main ideas involved in such a proof at the end of this section. The function $K$ is given by
\begin{align}
\begin{split}
K(\beta_S,\beta,&\rho_R,H_R,H_S)=\\
&\max \left\lbrace0,-\delta(\beta_S) \min_{\alpha\leq \delta(\beta_S)} \frac{\partial^2}{\partial \alpha^2}S_\alpha(\rho_R\|\omega_{\beta}(H_R))\right\rbrace,\nonumber
\end{split}
\end{align}
where
\begin{equation}
\delta(\beta_S):=\log(Z_\beta)/\mc{V}_{\beta}(\omega_{\beta_S}(H_S),H_S)\geq 0,\quad Z_\beta = \tr(\e^{-\beta H_S}).\nonumber
\end{equation}
The bound \eqref{eq:main_result} applies for any possible (diagonal) resource state, however finding $K(\beta_S,\beta,\rho_R,H_R,H_S)$ involves a minimization, that although is feasible for low dimensional systems, might be an obstacle for practical purposes when dealing with large systems and for values of $\beta_S$ so that $K(\beta_S,\beta,\rho_R,H_R,H_S)$ cannot be neglected. That said, we will investigate kinds of resource systems $\rho_R,H_R$ for which  $K(\beta_S,\beta,\rho_R,H_R,H_S)= 0$. In those cases, the general sufficient condition given by \eqref{eq:main_result} taken together with the necessary condition \eqref{eq:thirdlawintro} will imply that that a \emph{necessary} and \emph{sufficient} condition is given simply by
\begin{align}
\mc V_\beta(\rho_R,H_R) \geq \mc V_\beta(\omega_{\beta_S}(H_S),H_S),
\end{align}
In particular, we will see in section~\ref{sec:thermal}, that this holds true for large classes of thermal non-equilibrium resources. Let us also note that $K(\beta_S,\beta,\rho_R,H_R,H_S)$, just as the vacancy, is additive over non-interacting and uncorrelated resources. We will use this property in the next section to investigate the setting of \emph{i.i.d.} resources.

In the result given above, we have focused on thermal target states. This is in fact not necessary. We show in the appendix~\ref{sec:app:epsilon} that a completely analogous result holds for states of the form
\begin{align}
\rho_\epsilon = (1-\epsilon) \ketbra{0}{0} + \epsilon \rho^\perp,\quad \epsilon \ll 1.
\end{align}
where $\rho^\perp$ is any density matrix which has full rank on the subspace orthogonal to the ground-state $\ket{0}$ and commutes with $H_S$.

\subsection{Sketch of the proof of Thm. \ref{thm:generaltheorem}}

As we have seen in the previous section, a set of sufficient conditions for a transition with catalytic thermal operations is given by the infinite set of inequalities of \eqref{eq:second_laws}. The main idea behind the proof is that when the target system is sufficiently cold ($\beta_S > \beta_{\text{cr}}$) it suffices to check the conditions \eqref{eq:second_laws} for very small $\alpha$. This follows from the fact that for $\beta_S > \beta_{\text{cr}}$ the r.h.s. of \eqref{eq:second_laws}, given by $S_{\alpha}(\rho_S\| \omega_{\beta}(H_S))$, rapidly saturates to its maximum value as we increase $\alpha$ and it is concave (see Fig. \ref{fig:renyifigure}). Given that one only needs to consider small values of $\alpha$ it is possible to make a Taylor expansion around $\alpha=0$ of $S_{\alpha}$ of the form
\begin{eqnarray}
\nonumber S_{\alpha}(\rho_R\| \omega_{\beta}(H_R))&\approx& S_{0}(\rho\|\omega_{\beta}(H_R))+ \frac{\partial S_{\alpha}(\rho\|\omega_{\beta}(H_R))}{\partial \alpha}\bigg|_{\alpha=0}\!\!\!\!\!\!\alpha\\ \label{eq:taylormaintext}
&+& k\; \alpha^2.
\end{eqnarray}
This reduces the infinite inequalities of \eqref{eq:secondlaws} to a single one that depends on the derivate of $S_{\alpha}$ and an error term $k$, that is related to the error term  appearing in Thm. \ref{thm:generaltheorem} denoted by $K$. This expansion can be further simplified noting that $S_{\alpha=0}(\rho\|\omega_{\beta}(H))=0$. The vacancy comes into play because of the identity
\begin{equation}
\frac{\partial S_{\alpha}(\rho\|\sigma)}{\partial \alpha}\bigg|_{\alpha=0}=S_{1}(\sigma\|\rho):=S(\sigma\|\rho),
\end{equation}
which inverts the arguments of the second term on the r.h.s. of \eqref{eq:taylormaintext}. Taking all these elements into account and accounting properly for the precision of the Taylor approximations we arrive at an inequality involving only the vacancy and a vanishing error term as determined by Thm. \ref{thm:generaltheorem}.

\begin{figure}
\includegraphics[width=0.45\linewidth]{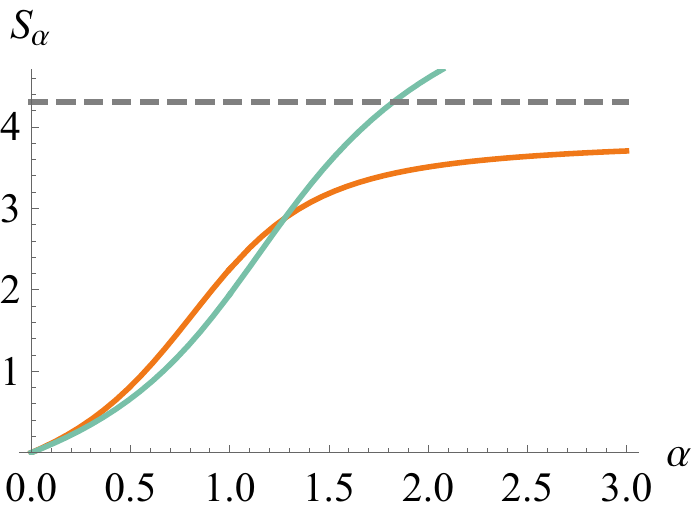}
\includegraphics[width=0.45\linewidth]{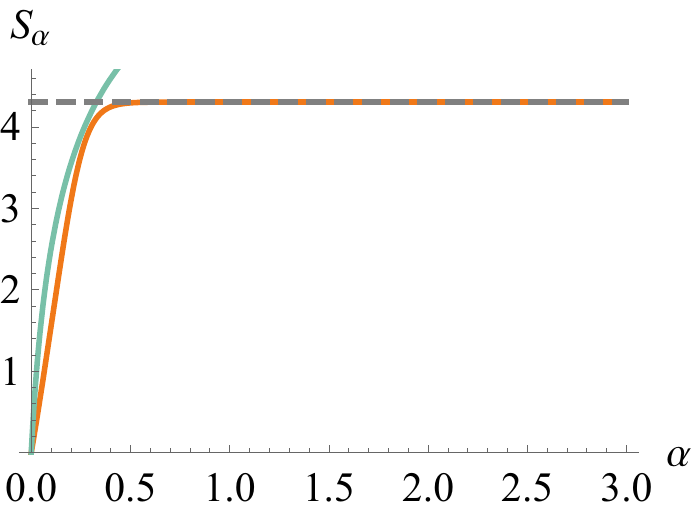}
\caption{The figure shows the behavior of $S_{\alpha}(\omega_{\beta_S}(H_S)\|\omega_{\beta}(H_S))$ (orange) and $S_{\alpha}(\rho_R\|\omega_{\beta}(H_R))$ (blue). \textbf{The left plot} shows a target state that is not very cold together with an insufficient resource. The transition is not possible since the blue line is below the orange line for $\alpha\lesssim 1.25$. \textbf{The right plot} shows the behavior when $\rho_S$ becomes very cold. The function becomes more similar to a step function. The fact that $S_{\alpha}(\rho_R\|\omega_{\beta}(H_R))$ (blue curve) is larger than the orange curve, which implies that the transition is possible, is determined by the behavior at very small values of $\alpha$, and up to a small error, by the fact that the derivate of the blue curve is larger than that of the orange curve at $\alpha=0$.}
\label{fig:renyifigure}
\end{figure}

\subsection{A short comment on catalysts}
\ro{As laid out in Sec. \ref{sec:set-up}, we define the catalytic thermal operations by including the possibility that the catalyst changes during the transition, as long as this change can be made arbitrarily small, as is standard in recent literature on the resource theoretic approach to thermodynamics \cite{Brandao2015}. This formulation is a form of exact catalysis, in the sense that as the error has to be arbitrarily small, the catalyst is returned arbitrarily unchanged. However, it is a possible to consider other forms of catalytic thermal operations which are either more restrictive about the change of the catalyst, where no error --not even arbitrarily small-- is allowed for; or less restrictive, in the sense the catalyst is allowed to change by a finite amount. We consider both alternatives on Appendix~\ref{app:sec:exact_catalysts} and \ref{app:approx_catalysts} respectively. First, we study the case in which one requires that the catalysts is always returned exactly in the same state, that is, taking $\epsilon=0$ in the definitions laid out in Sec. \ref{sec:set-up}. In this case it is no longer valid that a set of sufficient conditions is given by positive values of $\alpha$ on \eqref{eq:second_laws}, but one also has to consider Renyi-divergences for negative $\alpha$. This case is analyzed in Appendix~\ref{app:sec:exact_catalysts}, where we show that in this scenario i) The general necessary condition of Theorem~\ref{thm:necessary} holds and ii) A general sufficient condition similar to the one of Theorem~\ref{thm:generaltheorem} is derived. This general sufficient condition only differs on a multiplicative constant --independent of the final temperature to which one cools-- from the one derived in Theorem~\ref{thm:generaltheorem}.
Finally, in Appendix~\ref{app:approx_catalysts} we furthermore discuss the case of approximate catalysts. We put forward a consistent method to allow for finite errors on the catalyst while maintaining the validity of the third law of thermodynamics.
}

\section{I.i.d resources and scaling of the target temperature}
\label{sec:iid}
Theorem~\ref{thm:generaltheorem} together with the necessary condition~\eqref{eq:thirdlawnecessary} provide completely general sufficient and necessary condition, respectively, for cooling a system to target temperature $T_s=1/\beta_S$ (setting $k_B=1$) using a given resource $(\rho_R,H_R)$. They thus characterize the possibility of cooling in full generality. To obtain results for concrete physical situations and find out how the target temperature $T_S$ scales with physical key quantities of the resource --such as the system-size of the resource--  one has to choose a particular resource and calculate its vacancy as well as the error term $K$. Then one has to check how these quantities depend on the physical properties of interest.

We will now focus on the scaling between the size of the resource and the final temperature of the target system. For this we assume that the resource is given by a number of identically and independently distributed copies. Later we will also discuss other assumptions we can make about the resource.
 Thus, we consider the case where the resource state is given by $\rho_R=\varrho_R^{\otimes n}$ and Hamiltonian $H_R=\sum_i H_R^i$ where $H_R^i=\id_1 \otimes \cdots \otimes \id_{i-1} \otimes h_R \otimes \id_{i+1} \otimes \cdots \otimes \id_n$.

 Let us now consider the following task: Given fixed $\varrho_R,H_R,\beta,H_S$, find the minimum $n$ so that it is possible to cool down the target state to inverse temperature $\beta_S$.

By using Thm. \ref{thm:necessary} together with additivity of $\mc{V}_{\beta}$ we obtain that the \emph{necessary} number of copies $n^{\rm nec}(\beta_S)$ fulfills
\begin{equation}\label{eq:iidlowerbound}
n^{\rm nec}(\beta_S) \geq \frac{\mc{V}_\beta (\omega_{\beta_S}(H_S),H_S)}{\mc{V}_{\beta}(\varrho_R,h_R) }.
\end{equation}
By using Thm. \ref{thm:generaltheorem} in the \emph{i.i.d.} we also obtain a \emph{sufficient} number of copies $n^{\rm suff}$. The condition \eqref{eq:main_result} takes the form
\begin{equation}
n [ \mc{V}_{\beta}(\varrho_R,h_R) - K(\beta_S,\beta,\varrho_R,h_R,H_S) ]\geq \mc{V}_\beta (\omega_{\beta_S}(H_S),H_S).\nonumber
\end{equation}
Since this condition is sufficient, but not always necessary, we obtain
\begin{equation}\label{eq:iidupperbound}
n^{\rm suff}(\beta_S) \leq \frac{\mc{V}_\beta (\omega_{\beta_S}(H_S),H_S)}{\mc{V}_{\beta}(\varrho_R,h_R) - K(\beta_S,\beta,\varrho_R,h_R,H_S) }.
\end{equation}
Since the correction $K$ goes to zero as $\beta_S\rightarrow \infty$ (the target temperature going to zero), we see that
\begin{align}
\lim_{\beta_S\rightarrow \infty} \frac{n^{\rm suff}(\beta_S)}{n^{\rm nec}(\beta_S)} = 1.
\end{align}
It is also interesting to re-express the previous conditions to obtain a more transparent relation between the final achievable temperature and the number of copies. We will see next that $n^{\rm nec}(\beta_S)$ and $n^{\rm suff}(\beta_S)$ scale as $\beta_S$ for large $\beta_S$. Thus the target temperature approaches zero as $1/n$.

\subsection{Scaling of the target temperature}\label{sec:vac_free_energies}

In the special case where the target state is a thermal state one can reformulate the vacancy in terms of non-equilibrium free energies. Indeed, the vacancy of a thermal state at temperature $\beta_S$ simply takes the form
\begin{align}\label{eq:vac_as_F}
\vac_\beta(\omega_{\beta_S}(H),H) = \beta_S \Delta F_{\beta_S}(\omega_\beta(H_S),H_S),
\end{align}
with $\Delta F_\beta(\rho,H) = \langle H\rangle_\rho - \langle H \rangle_\beta - (S(\rho)-S(\omega_\beta))/\beta$.

In this case the condition \eqref{eq:thirdlawnecessary} reads:
\begin{align}
 \mc V_\beta(\rho_R,H_R) \geq \beta_S \Delta F_{\beta_S}(\omega_\beta(H_S),H_S).
\end{align}

From \eqref{eq:vac_as_F} we see that for large $\beta_S$ we have (assuming vanishing ground-state energy)
\begin{align}
\vac_\beta(\omega_{\beta_S}(H),H) = \beta_S E_\beta-S_\beta,\quad \text{as}\ \ \beta_S\rightarrow \infty.
\end{align}
Assuming again a resource system of $n$ non-interacting identical particles each described by $(\varrho_R,h_R)$, we then obtain that the minimum achievable temperature $T^{(n)}_S$ scales as
\begin{align}
 T^{(n)}_S  = \frac{1}{n}\frac{E^S_\beta}{\mc V_\beta(\varrho_R,h_R)},\quad n \gg 1.
\end{align}
This result is similar to the asymptotic result of Janzing et al. \cite{Janzing00}.

Lastly, let us point out that the above scaling relation implies that the probability $p$ to find the system in the ground-state after the cooling procedure increases exponentially to $1$ with $n$. For example, if the target system is a $d+1$ dimensional system with gap $\Delta$ above a unique ground-state, we have (for large $n$)
\begin{align}
p \geq 1/(1+d\e^{-\beta_S \Delta}) \approx 1 - d \e^{-n \mc V_\beta(\omega_{\beta_R}(h_R),h_R) \frac{\Delta}{E^S_\beta}}.\nonumber
\end{align}
Thus, while an exact third law holds in the sense that $n\rightarrow \infty$ for $T_S\rightarrow 0$, the ground-state probability asymptotically converges very quickly to unity.

The above relations demonstrate how one can obtain quantitative expressions of the unattainability principle from Theorems~\ref{thm:necessary} and \ref{thm:generaltheorem} by making assumptions about the given resources.

\subsection{Scaling of the vacancy with system size}
In the case of i.i.d. resources with non-interacting Hamiltonians the vacancy is an extensive quantity in the sense that it scales linearly with the number of particles. However, for arbitrary quantum systems with correlated and interacting constituents, it is in general difficult to calculate the vacancy and hence estimate directly how it scales with the number of particles. Nonetheless,  we can use the relation~\eqref{eq:vac_as_F} to argue that the vacancy will be extensive for large classes of many-body systems.

In particular, let us assume that the resource is in a thermal state of some local Hamiltonian. That is $\rho_R=\omega_{\tilde{\beta}}(\tilde{H}_R)$, where $\tilde{H}_R$ is \emph{any} local Hamiltonian (possibly differing form $H_R$) and $\tilde{\beta}$ is finite. In this case one can use Eq. \eqref{eq:vac_as_F} to write
\begin{equation}\label{eq:vancancy_free}
\mc V_\beta(\rho_R,H_R)= \tilde{\beta}\Delta F_{\tilde{\beta}}(\omega_{\beta_R}(H_R),\tilde{H}_R).
\end{equation}
From the fact that the von~Neumann entropy is sub-additive and from the locality of the Hamiltonians $H_R$ and $\tilde{H}_R$ it then follows that the vacancy $\mc V_\beta(\rho_R,H_R)$ scales (at most) \emph{linearly} with the system size. As a consequence, the minimal final temperature $T_S^{(n)}$ scales (at best) inversely proportional to the volume of the resource.

In the light of the previous considerations, it seems likely that a similar scaling holds for any resource (potentially under reasonable physical assumptions, such as clustering of correlations). We leave the general characterisation of many-body systems such that the vacancy is extensive as an interesting future research direction.

\section{Thermal resources}
\label{sec:thermal}
As discussed after the statement of Theorem~\ref{thm:generaltheorem}, it is useful to find general conditions under which the error term $K$ disappears and the sufficient condition coincides with the general necessary condition. Naturally, it is necessary to make additional assumptions about the resources to achieve this.

We now consider as resource state $\rho_R$ a (possibly multi-partite) thermal state of some Hamiltonian $H_R$ at inverse temperature $\beta_R$. In the following we will derive a simple expression that allows us to check whether
\begin{equation}\label{eq:positivityerror}
K(\beta_S,\beta,\rho_R,H_R,H_S)= 0, \:\: \forall \beta_S>0,H_S,
\end{equation}
and hence \eqref{eq:thirdlawnecessary} becomes a necessary and sufficient condition.
The reasoning is based on showing that
\begin{equation}\label{eq:renyientropythermal}
S_\alpha(\rho_R=\omega_{\beta_R}(H_R)||\omega_{\beta}(H_R))
\end{equation} is convex for a given range of $\alpha<1$, which implies \eqref{eq:positivityerror}. The convexity of \eqref{eq:renyientropythermal} can be determined by looking at the convexity of the average energy as a function of the inverse temperature of the resource
\begin{equation}
x \mapsto E_{x}^R := \tr (\omega_{x}(H_R)H_R).
\end{equation}
In particular, we will see that if $\beta_R< \beta$ and the function $x\mapsto E^R_x$ is convex  for $x\in[\beta_R,\beta]$, then the function $S_\alpha(\omega_{\beta_R}(H_R)||\omega_{\beta}(H_R))$ is convex for all $\alpha<1$.

\begin{theorem}\label{cor:simplecorollary} For resources of the form $(\omega_{\beta_R}(H_R),H_R)$ that are hotter than the bath, that is with $\beta_R \leq \beta$, if $E^R_{x}=\tr(\omega_{x}(H_R) H_R)$ is convex in the range $x\in [\beta_R,\beta]$, then \eqref{eq:thirdlawnecessary} is a sufficient and necessary condition for low temperature cooling.
\end{theorem}
This Theorem simplifies considerably the task of formulating bounds on the third law, since the average energy is a much more accessible quantity than the Renyi-divergences. In Sec. \ref{sec:systemsforwhichisconvex} we discuss several classes of physically motivated conditions that imply that $E^R_{x}=\tr(\omega_{x}(H_R) H_R)$ is convex. We emphasize, however, that the convexity of the energy is not a necessary condition for the correction $K$ to vanish: There are cases for which $x \mapsto E^R_x$ is not convex for the whole range of inverse temperatures $[\beta_R,\beta]$ and \eqref{eq:thirdlawnecessary} is nevertheless a sufficient and necessary condition for cooling.

Lastly, let us mention that condition \eqref{eq:positivityerror} is fulfilled if, for a fixed $\beta_R$, the bath's inverse temperature $\beta$ is sufficiently large, without any extra assumption on the convexity of $E_x$. This implies, that for sufficiently cold baths, \eqref{eq:thirdlawnecessary} is also a sufficient and necessary condition. This is shown in Appendix \ref{sec:app:thermalstates}, together with several properties of the Renyi divergences for thermal states that might be of independent interest and also include a proof of Thm. \ref{cor:simplecorollary}.

\subsection{Systems for which the energy is convex}

\label{sec:systemsforwhichisconvex}
As implied by Cor. \ref{cor:simplecorollary}, Eq.~\eqref{eq:thirdlawnecessary} becomes a sufficient and necessary condition for cooling if the resource is a thermal state hotter than the bath and its average energy is convex in the inverse temperature. We will now see that the convexity of the energy is fulfilled by a wide range of physical models.

We will first re-express the convexity in terms of the heat capacity. This allows to check for general systems whether $E^R_x$ is convex, as the heat capacity as a function of the temperature is an intensively studied quantity for many-body systems. Using the definition of heat capacity $C_{x}:=\frac{ \mathrm{d}  E^R_{x}}{\mathrm{d}T}$, with $T=1/x$, we find that the convexity of the energy, as formulated in the condition of Thm. \ref{cor:simplecorollary}, can be expressed as
\begin{equation}\label{eq:conditionconvexity}
\frac{ \mathrm{d}^2  E^R_{x}}{\mathrm{d}x^2}=\frac{1}{x^2}\left( \frac{C_{x}}{x} - \frac{\mathrm{d} C_x}{\mathrm{d}x}\right)\geq 0  \text{ with } x\in [\beta_R,\beta].
\end{equation}
Equivalently, this condition can be expressed as
\begin{align}
\frac{1}{\beta_R^2} C_{\beta_R} - \frac{1}{\beta'^2} C_{\beta'}\geq  0
\end{align}
for all $\beta_R\leq \beta' \leq \beta$.
In most thermodynamics systems, the heat capacity is monotonically increasing with the temperature, hence $\frac{\mathrm{d} C_x}{\mathrm{d}x} \leq 0$ and \eqref{eq:conditionconvexity} is satisfied. A seminal exception to this case is given by the so-called Schottky anomaly, which is present in certain solid states at very low temperatures \cite{Tari2003}.
We thus see, that for thermodynamic systems the convexity of the energy is a very natural property.
Nevertheless it can fail, in particular in finite systems. We will now show that even for large classes of finite systems the energy is convex.

This is due to the following Lemma, which we prove in the appendix.
\begin{lemma}[Equidistant levels]\label{lemm:equidistant}
Consider any Hamiltonian with equidistant and non-degenerate energy levels. Then the function $\beta\mapsto E_\beta$ is convex.
\end{lemma}
Immediate examples of Hamiltonians with equidistant energy levels are two-level systems or harmonic oscillators. But in fact, the lemma covers a much wider class of models since the vacancy is unitarily invariant and additive over non-interacting subsystems.

 It follows that also any harmonic system and any system described by free fermions has a convex energy function, since free bosonic and fermionic systems can always be made non-interacting by a normal-mode decomposition. In such a normal-mode decomposition they simply correspond to a collection of non-interacting harmonic oscillators or two-level systems, respectively.

  These systems include highly-correlated (even entangled) systems and no thermodynamic limit needs to be taken. A particularly interesting resource that is included by these results is that of hot thermal light, which has been considered before as a valuable resource for cooling \cite{Mari2012}.

It can furthermore be checked that for large but finite many-body systems whose density of states in the bulk is well approximated by $\mu(\epsilon) \simeq \e^{\gamma \epsilon - \alpha \epsilon^2}$ the average energy $E_\beta$ is convex in $\beta$ \footnote{Such a density of states is typical for many-body systems in the bulk of the energy, but deviations typically do appear in the tails of the distribution.}.

Finally, for every finite system there is a critical $\beta_c$ such that $E_\beta$ is convex for all $\beta>\beta_c$. Thus as soon as $\beta>\beta_R>\beta_c$, the sufficient condition~\eqref{eq:thirdlawnecessary} holds for small enough target temperatures. This means that if an experimenter has a mechanism to pre-cool the environment to a very low temperature, well below $1/\beta_c$, and the resource has a temperature, which is larger than that of the environment but still smaller than $1/\beta_c$, then condition~\eqref{eq:thirdlawnecessary} holds as a sufficient and necessary condition.

\subsection{A source of work}

Our formalism can also incorporate a source of work as particular case of a resource for cooling. The limitations on cooling as a function of the input of work have been studied in Ref. \cite{Masanes2017}. There it is shown that the fluctuations of work, rather than its average value, have to diverge when the target state reaches vanishing final temperature and if the heat bath has finite heat capacity. We will here derive a result similar in spirit using Thm. \ref{thm:generaltheorem}, although we employ a different model for the work source. Importantly, our result implicitly allows for infinite heat capacity in the heat bath. It should thus be viewed as being complementary to the results in Ref.~\cite{Masanes2017}.

Let us model the work source by a system $R$ with Hamiltonian
\begin{equation}
H^w_R=\sum_{x=-d/2}^{d/2}\frac{Ex}{d}\ketbra{\frac{Ex}{d}}{\frac{Ex}{d}}.
\end{equation}
One can see this Hamiltonian as a $d$-dimensional harmonic oscillators with energies bounded between $E/2$ and $-E/2$. We are interested in the limit of $d\rightarrow \infty$. In this case $R$ is similar to the model put forward in Ref. \cite{Skrzypczyk2011,Masanes2017}, with the difference that we consider here a finite value of $E$. We enforce the condition that the battery acts as a energy reservoir and not as an entropy sink (that would make the task of cooling trivial) by assuming the work source to be in state $\rho^w_R=\id/d$ (we can also interpret this as the work-source being at temperature $+\infty$). These assumptions on the work source are justified by the fact that it fulfills the second law of thermodynamics.

By this statement we mean the following. Suppose we want to use a non-equilibrium state $\rho$ of some system $S$ with Hamiltonian $H_S$ to extract work and put it as average energy into the work-source. To do this we implement a (catalytic) thermal operation on the heat bath, system $S$ and the work-source $R$. Then the increase of energy on the work-source (i.e., the work) $\Delta E_R$ is bounded by the non-equilibrium free energy of the \emph{system} as
\begin{align}
  \Delta E_R \leq \Delta F_\beta(\rho,H_S).
\end{align}
This is shown in Appendix \ref{sec:worksource}.

Now we will show that the third law can be obtained, in the sense that both $E$ and $d$ have to diverge in order to be able to use $R$ to cool down a system to zero final temperature. Let us first recall that by Lemma \ref{lemm:equidistant} a sufficient and necessary condition for cooling for such a resource is given by \eqref{eq:thirdlawnecessary}. Furthermore, the vacancy of the work source is given by
\begin{eqnarray}
\nonumber \mc{V}_{\beta}(\rho_R^w,H_R^w)&=&S(\omega_{\beta}(H_R^w)\| \id/d)\\
\nonumber&=&\tr \left( \omega_{\beta}(H_R^w)(\log (\omega_{\beta}(H_R^w) - \log (\id/d))\right)\\
\nonumber &=& -\beta \tr \left( \omega_{\beta}(H_R^w)H_R^w\right) + \log(d)
- \log(Z_{\beta} (H_R^w))\\
&\leq & -\beta E/2 +\log(d) - \log(Z_{\beta} (H_R^w)).
\end{eqnarray}
The partition function can be upper bounded as
\begin{eqnarray}
Z_{\beta}(H_R^w)&=& \sum_{x=-d/2}^{d/2} e^{-\beta \frac{Ex}{d}}\\
&\geq & e^{\beta E/2} + (d-1) e^{-\beta E/2}\\
&=& e^{\beta E/2} d\left( \frac{1}{d} + \frac{d-1}{d} e^{-\beta E}\right).
\end{eqnarray}
Hence, we find that
\begin{equation}
\mc{V}_{\beta} (\rho_R^w,H_R^w) \leq - \log \left( \frac{1}{d} + \frac{d-1}{d} e^{-\beta E}\right).
\end{equation}
Combined with \eqref{eq:thirdlawnecessary}, this implies that a necessary condition for cooling to a state $\rho_S$ is given by
\begin{equation}\label{eq:thirdlawwork}
\mc{V}_\beta ( \rho_S,H_S) \leq - \log \left( \frac{1}{d} + \frac{d-1}{d} e^{-\beta E}\right).
\end{equation}
Most importantly, note that in order to obtain a state $\rho_S$ that is close to a Gibbs state at zero temperature the r.h.s. of \eqref{eq:thirdlawwork} has to diverge. For this to be possible both $E$ and $d$ have to diverge, since
\begin{eqnarray}
\lim_{d\rightarrow \infty} \mc{V}_{\beta}(\rho_R^w,H_R^w) &\leq& \beta E,\\
\lim_{E\rightarrow \infty} \mc{V}_{\beta}(\rho_R^w,H_R^w) &\leq& \log (d).
\end{eqnarray}
This implies the unattainability principle in the sense that an infinitely dense spectrum with unbounded energy is needed for cooling to absolute zero.

Before coming to the proof of Thm.~\ref{thm:generaltheorem} and our conclusions, let us briefly comment on a different model of work and a possible source of confusion that might arise. A  model of work known as a \emph{work-bit} has been used in the literature of thermal operations \cite{Horodecki2013,Brandao2015}. In this model, it is assumed that the work-source is a two-level system with energy-gap $W$ that undergoes a transition from the excited state $\ket{W}$ to the ground-state $\ket{0}$ to implement a transition on a system $S$. Using the results of Ref. \cite{Horodecki2013} one can show that it is possible to cool a system to the ground-state in this model as long as $W> \log Z_\beta$, where $Z_\beta$ is the partition function of the system $S$. This means that there exists a (catalytic) thermal operation that implements cooling in the sense that
\begin{align}
\omega_\beta(H_S)\otimes\proj{W} \mapsto \proj{0}\otimes \proj{0}. \label{eq:cooling_wit}
\end{align}
At first this seems to be in conflict with our results. However, using the vacancy, it is easy to show that the above process is extremely unstable: it only works for pure initial states of the work-bit. Indeed, one can use condition \eqref{eq:thirdlawnecessary} to establish limits on cooling if the initial state on the work-bit is any full-rank state that approximates the excited state $\proj{W}$ to \emph{arbitrary} but finite precision. A simple calculation for a initial state of the work system as $\rho=(1-\epsilon)\proj{W} +\epsilon\proj{0}$ with Hamiltonian $H_W=W \proj{W}$ yields the bound
\begin{equation}
\mc{V}_{\beta}(\rho,H_W)\leq - \log\left( \epsilon(1-\epsilon)\right) \:\: \forall \: W.
\end{equation}
This implies that perfect cooling is impossible for any value of $W$ --even diverging-- if the initial state of the work-bit is a full-rank state.

\section{Proof of Thm. \ref{thm:generaltheorem}}\label{sec:proofmaintheorem}
We will now proof Thm.~\ref{thm:generaltheorem}. Before we go to the details, let us first  explain the general logic behind the result. It is clear that to obtain a single necessary and sufficient condition for cooling at low temperatures, we have to show that the infinite set of second laws in eq.~\eqref{eq:second_laws} collapse to a single condition.
The first important step in the proof is the following Lemma.
\begin{lemma}[Concavity at low temperatures]\label{lemma:concavity}
Let $\beta>0$ and a Hamiltonian $H_S$ be given. There exists a critical inverse temperature $\beta_{\rm cr}$ such that for all $\beta_S>\beta_{\rm cr}$ and
\begin{align}
\alpha \mapsto S''_\alpha(\omega_{\beta_S}(H_S)\|\omega_\beta(H_S)) \leq 0.
\end{align}
and
\begin{align}
S_\infty (\omega_{\beta_S}(H_S)\|\omega_\beta(H_S)) \leq \log Z_\beta.
\end{align}
Here, the critical value $\delta(\beta_S)$ is given by
\begin{align}
\delta(\beta_S) = \frac{\log (Z_\beta)}{\vac_\beta(\omega_{\beta_S}(H_S),H_S)}.
\end{align}
\begin{proof}
See appendix \ref{app:sec:concavity_low_temperatures}.
\end{proof}
\end{lemma}
 Using this result, we can now \emph{upper bound} the Renyi-divergence on the target by its linear approximation at the origin in this parameter regime.
Since $S'_0(\rho || \omega_\beta(H)) = S(\omega_\beta(H) || \rho)=\vac_\beta(\rho,H)$, we get
\begin{align}
S_\alpha(\omega_{\beta_S}(H_S)||\omega_\beta(H_S)) \leq \vac_\beta(\omega_{\beta_S}(H_S),H_S)\alpha,\quad \forall \alpha\leq \alpha_c.
\end{align}
Secondly, for small enough target temperatures we also have $S_\infty(\omega_{\beta_S}(H_S)||\omega_\beta(H_S)) \leq \vac_\beta(\omega_{\beta_S}(H_S),H_S)\alpha_c$. Since $\alpha\mapsto S_\alpha$ is monotonously increasing, the second laws in eq.~\eqref{eq:second_laws} are hence also satisfied if
\begin{align}
S_\alpha(\rho_R || \omega_\beta(H_R)) > \vac_\beta(\omega_{\beta_S}(H_S),H_S)\alpha,\quad \forall\alpha\leq \alpha_c.
\end{align}
For small temperatures, we can further restrict the range of $\alpha$ to the interval $[0,\delta(\beta_S))$, where $\delta(\beta_S)$ is given by:
\begin{align}
\delta(\beta_S) = \frac{S_\infty(\omega_{\beta_S}(H_S)||\omega_\beta(H_S))}{\vac_\beta(\omega_{\beta_S}(H_S),H_S)}.
\end{align}

The final step is now given by bounding the Renyi-divergence of the \emph{resource} $S_\alpha(\rho_R || \omega_\beta(H_R))$. In particular, if we knew that it was convex (such as in the case of a thermal resource with $E^R_\beta$ being convex), we could \emph{lower} bound it by its linear approximation at the origin and obtain the necessary and sufficient condition \eqref{eq:thirdlawnecessary}.

In the general case, $S_\alpha(\rho_R || \omega_\beta(H_R))$ will not be convex. But we can use that we only have to check small values of $\alpha<\delta(\beta_S)$ and simply Taylor expand $S_\alpha(\rho_R || \omega_\beta(H_R))$. Using Taylor's theorem we then obtain
\begin{align}
S_\alpha(\rho_R || \omega_\beta(H_R)) \geq \vac_\beta(\rho_R,H_R)\alpha - k(\beta_S,\beta,\rho_R,H_R)\alpha^2.
\end{align}
This yields as new \emph{sufficient} condition
\begin{align}
\vac_\beta(\rho_R,H_R)\alpha - k(\beta_S,\beta,\rho_R,H_R)\alpha^2 \geq \vac_\beta(\omega_{\beta_S}(H_S),H_S)\alpha, \nonumber
\end{align}
for all $0<\alpha\leq \delta(\beta_S)$. The function $k(\beta_S,\beta,\rho_R,H_R)\geq 0$ is given by
\begin{align}
k(\beta_S,\beta,\rho_R,H_R) = \max\left\{0,-\min_{\alpha\leq \delta(\beta_S)} S''_\alpha(\rho_R \| \omega_\beta(H_R))\right\}.\nonumber
\end{align}
We can now divide the sufficient condition by $\alpha$ and, since $k(\beta_S,\beta,\rho_R,H_R)\geq 0$, replace $\alpha$ by $\delta(\beta_S)$ to arrive at the final sufficient condition
\begin{align}
\vac_\beta(\rho_R,H_R) - K(\beta_S,\beta,\rho_R,H_R,H_S) \geq \vac_\beta(\omega_{\beta_S}(H_S),H_S), \nonumber
\end{align}
with $K(\beta_S,\beta,\rho_R,H_R,H_S) = k(\beta_S,\beta,\rho_R,H_R)\delta(\beta_S)$. This finishes the proof. \qed

\section{Summary}
In this work we have investigated the limits on low temperature cooling when arbitrary systems out of equilibrium are used as a resource. We provide sufficient and necessary conditions that establish novel upper and lower bounds on the amount of resources that are needed to cool a system close to its ground state. We found that the limitations are ruled by a single quantity, namely the vacancy. This is remarkable, since at higher temperatures there is an infinite family of ``second laws'' that need to be checked to determine whether a non-equilibrium state transition is possible.

We have only focused on the amount of non-equilibrium resources, as we assume access to an infinite heat bath and we leave considerations about the time and complexity of the cooling protocol aside. These other kind of resources have been explored in other complementary works on the third law \cite{Masanes2017,Scharlau2016,Silva2016,DiFranco2013}. It would be interesting to see if the the vacancy plays a role to express the limitations on the size of the heat bath or any other resources that diverge when cooling a system to absolute zero. More particularly, it is an interesting question for future work to obtain the optimal sufficient scaling of the size of the heat bath and the potential ``catalyst'' $\tau$ that is needed to cool the system to the final low temperature \cite{Ng2015}.
In this work, we have required the catalyst to be returned exactly. The necessary condition~\eqref{eq:thirdlawnecessary} and the resulting quantitative unattainability principle is, however, stable when one requires instead that the vacancy of the catalyst only changes little (see appendix \ref{app:approx_catalysts} for a discussion of approximate catalysts). We leave it as an open problem to study how the sufficient condition behaves in such an approximate scenario.

The results of Sec. \ref{sec:thermal} suggests that for a large class of physically relevant systems the third law can be expressed simply as the monotonicity of the vacancy. It would be of interest to specify more general assumptions on a many-body system so that this is the case. On the other hand, there exist systems for which the vacancy is not a sufficient condition. This open the possibility to have families of resources that, although are out of equilibrium, are useless for cooling. We leave this as an open question for future work.

Lastly, we note that in this work we have focused on the expenditure of non-equilibrium resources for low temperature cooling, which are precisely the resources that are employed in laser cooling \cite{lasercooling}. We leave as an interesting open research direction to analyse protocols of laser cooling in the light of the bounds obtained here.

\emph{Acknowledgments.}
We thank Lluis Masanes, Joseph M. Renes and Jens Eisert for interesting discussions and valuable feedback. This work has been supported by the ERC (TAQ), the DFG (GA 2184/2-1) and the Studienstiftung des Deutschen Volkes.
\bibliography{cooling.bib}
\clearpage

\onecolumngrid
\appendix

\section{Proof of concavity of Renyi divergence for low temperatures}
\label{app:sec:concavity_low_temperatures}
In this section we proof the Lemma~\ref{lemma:concavity} about concavity of the Renyi-divergence at low temperatures. The Lemma holds for any Hamiltonian with pure point spectrum, a gap above the ground-state and  the property that the partition sum exists for any positive temperature.
\begin{lemma}[Concavity at low temperatures]
Let $\beta>0$ and a Hamiltonian $H_S$ with ground-state degeneracy $g_0$ be given. There exists a critical inverse temperature $\beta_{\rm cr}$ such that for all $\beta_S>\beta_{\rm cr}$ and for all $0<\alpha<\delta(\beta_S)$ we have
\begin{align}
\alpha \mapsto S''_\alpha(\omega_{\beta_S}\|\omega_\beta) \leq 0.
\end{align}
and
\begin{align}
S_\infty (\omega_{\beta_S}\|\omega_\beta) \leq \log Z_\beta.
\end{align}
Here, the critical value $\delta(\beta_S)$ is given by
\begin{align}
\delta(\beta_S) = \frac{\log (Z_\beta)}{\vac_\beta(\omega_{\beta_S}(H_S),H_S)} <1.
\end{align}
\end{lemma}

\begin{proof}
Let us first prove that the max-Renyi-divergence is upper bounded by the partition function at the environment temperature. Suppose that $\beta_S>\beta$ and let us write
\begin{align}
S_{\alpha}(\omega_{\beta_S}(H_S)\|\omega_\beta(H_S))&=\frac{1}{\alpha-1}\log\left(\sum_ig_i \e^{-\alpha(\beta_S-\beta)E_i}(Z_{\beta}/Z_{\beta_S})^\alpha \frac{\e^{-\beta E_i}}{Z_\beta}\right) \\
&= \frac{1}{\alpha-1}\log\left(\e^{-\alpha(\beta_S-\beta)E_0}(Z_{\beta}/Z_{\beta_S})^\alpha\sum_ig_i\e^{-\alpha(\beta_S-\beta)(E_i-E_0)}\frac{\e^{-\beta E_i}}{Z_\beta}\right),
\end{align}
where $E_i$ denote the different energies of $H_S$, with degeneracies $g_i$. Assuming w.l.o.g. $E_0=0$, we write this as
\begin{align}
S_{\alpha}(\omega_{\beta_S}(H_S)\|\omega_\beta(H_S))&= \frac{\alpha}{\alpha-1}\log( Z_\beta/Z_{\beta_S}) + \frac{1}{\alpha-1}\log\left(1 + \sum_{i>0} \e^{-\alpha(\beta_S-\beta)E_i}\frac{\e^{-\beta E_i}}{Z_\beta} \right).
\end{align}
It is now obvious that in the limit we obtain
\begin{align}
S_{\infty}(\omega_{\beta_S}(H_S)\|\omega_\beta(H_S))&= \lim_{\alpha\rightarrow\infty}S_{\alpha}(\omega_{\beta_S}(H_S)\|\omega_\beta(H_S)) = \log(Z_{\beta}) - \log(Z_{\beta_S}) \leq \log Z_{\beta}.
\end{align}

As a second step let us find the condition for which $\delta(\beta_S)<1$. To do that we express the vacancy as
\begin{align}
\vac_\beta(\omega_{\beta_S}(H_S),H_S) = \beta_S E_\beta - S_\beta + \log Z_{\beta_S},
\end{align}
where we write $S_\beta := S(\omega_\beta(H_S))$. We thus need that
\begin{align}
\beta_S E_\beta - S(\omega_\beta(H_S)) > \log Z_\beta - \log Z_{\beta_S}.
\end{align}
Relaxing to the sufficient criterion $\beta_S E_\beta - S_\beta > \log Z_\beta = S_\beta - \beta E_\beta$ we thus obtain
\begin{align}
\beta_S > \frac{2S_\beta - \beta E_\beta}{E_\beta}.
\end{align}

Let us now turn to the concavity. We will use the representation of $S''_\alpha$ proven in the next section, which is given by
\begin{align}
S''_\alpha(\omega_{\beta_S}\|\omega_\beta) = \frac{2}{(1-\alpha)^3}\left(\log Z_{\beta_S}-\log Z_{\tilde\beta(\alpha)} + (\beta_S-\tilde{\beta}(\alpha))E_{\tilde\beta(\alpha)}-(\beta_S-\tilde\beta(\alpha))^2 \mathrm{Var}(H)_{\tilde\beta(\alpha)}\right),
\end{align}
where $\tilde\beta(\alpha)=\beta(1-\alpha)+\alpha \beta_S$.  Since we are only interested in $\alpha<\delta(\beta_S)<1$, we have $\beta\leq \tilde\beta(\alpha) < \beta_S$. We therefore have to show that the terms in the parenthesis are negative. Let us use that the average energy is monotonic with $\beta$ and that $Z_{\tilde\beta(\alpha)}>1$ to bound these terms as
\begin{align}
\label{eq:parenthesis}
\text{parenthesis}\ &\leq \log Z_{\beta_S} + (\beta_S-\tilde{\beta}(\alpha))E_{\tilde\beta(\alpha)}-(\beta_S-\tilde\beta(\alpha))^2 \mathrm{Var}(H)_{\tilde\beta(\alpha)}\nonumber\\
&\leq \log Z_{\beta_S} + (\beta_S-\beta)E_{\beta}-(\beta_S-\tilde\beta(\alpha))^2 \mathrm{Var}(H)_{\tilde\beta(\alpha)}\nonumber\\
&\leq \log(d) + (\beta_S-\beta)E_{\rm max}-(\beta_S-\tilde\beta(\alpha))^2 \min_{x\in [\beta,\tilde\beta(\alpha)]}\mathrm{Var}(H)_{x}.
\end{align}
Now we bound $\tilde\beta(\alpha)$.  To do that we use that $\tilde\beta(\alpha)\leq \tilde\beta(\delta(\beta_S))=: \tilde\beta^*(\beta_S)$. It is clear that if we can bound $\tilde\beta^*(\beta_S)$ by a constant, the terms in the parenthesis become negative for some $\beta_S$ since the second order term in $\beta_S$ dominates. To see that $\tilde\beta^*(\beta_S)$ is indeed upper bounded by a constant, we again write the vacancy as
\begin{align}
\vac_\beta(\omega_{\beta_S}(H_S),H_S) &= -S(\omega_\beta)+ \beta_S E_\beta + \log Z_{\beta_S}
\end{align}
to obtain
\begin{align}
\beta^*:= \lim_{\beta_S\rightarrow \infty} \tilde\beta^*(\beta_S) &= \lim_{\beta_S\rightarrow \infty}\beta(1-\delta(\beta_S)) + \delta(\beta_S)\beta_S\\
&= \beta + \lim_{\beta_S\rightarrow \infty}\frac{\log Z_\beta}{\beta_S E_\beta +\log Z_{\beta_S}-S(\omega_\beta(H_S))}\beta_S = \beta+ \frac{\log Z_\beta}{E_\beta}.
\end{align}
This finishes the proof that $\beta_{\rm cr}$ as claimed in the Lemma exists. We also note that the function $\tilde\beta^*(\beta_S)$ is monotonically decreasing for all $\beta_S$ such that $\tilde\beta^*(\beta_S)<1$.
Finally, note that eq.~\eqref{eq:parenthesis} allows to give upper bounds \ro{on $\beta_{\rm cr}$} once one has lower bounds on the energy variance for inverse temperatures in the interval $[\beta,\beta^*]$.
\end{proof}


\section{Renyi divergence between thermal states}\label{sec:app:thermalstates}
Here, we will specialize to the situation where the resource states are thermal, with inverse temperature $\beta_R$. We now calculate the Renyi divergence for $\alpha<1$ in this case. We first write:
\begin{align}
S_\alpha (\omega_{\beta_R} || \omega_\beta) &= -\frac{\alpha}{\alpha-1}\log Z_{\beta_R} + \log Z_\beta + \frac{1}{\alpha-1}\log\tr(\e^{-\beta_R H\alpha}\e^{-\beta H(1-\alpha)}) \\
&= -\frac{\alpha}{\alpha-1}\log Z_{\beta_R} + \log Z_{\beta} + \frac{1}{\alpha-1}\log Z_{(\beta_R-\beta)\alpha + \beta}\\
&= -\frac{\alpha-1}{\alpha-1}\log Z_{\beta_R} + \log Z_{\beta} + \frac{1}{\alpha-1}\log(Z_{(\beta_R-\beta)\alpha + \beta}/Z_{\beta_R}) \\
&= \log(Z_{\beta}/Z_{\beta_R}) + \frac{1}{\alpha-1}\log(Z_{(\beta_R-\beta)\alpha + \beta}/Z_{\beta_R}).
\end{align}

We will now show that the function is convex provided $\beta_R<\beta$ and that the  function $x\mapsto E_{\beta_R+x}$ is convex for $0\leq x \leq \beta-\beta_R$.  For the second derivative (with $\tilde{\beta}=(\beta_R-\beta)\alpha + \beta$) we obtain:
\begin{align}
  S_\alpha (\omega_{\beta_R} || \omega_\beta)'' &= \frac{2}{(1-\alpha)^3}\log Z_{\beta_R} -\frac{2}{(1-\alpha)^3}\log Z_{\tilde{\beta}} - 2\frac{1}{(1-\alpha)^2}\partial_\alpha \log Z_{\tilde{\beta}} + \frac{1}{\alpha-1}\partial^2_\alpha \log Z_{\tilde{\beta}} \\
&= \frac{2}{(1-\alpha)^3}\log Z_{\beta_R} -\frac{2}{(1-\alpha)^3}\log Z_{\tilde{\beta}} - 2\frac{1}{(1-\alpha)^2}(\beta-\beta_R)E_{\tilde{\beta}} - \frac{1}{1-\alpha}(\beta-\beta_R)^2 \mathrm{Var}(H)_{\tilde{\beta}}  \\
&=\frac{2}{(1-\alpha)^3}\left[\log Z_{\beta_R} - \log Z_{\tilde{\beta}} -(1-\alpha)(\beta-\beta_R)E_{\tilde{\beta}} - \frac{(1-\alpha)^2}{2}(\beta-\beta_R)^2\mathrm{Var}(H)_{\tilde{\beta}} \right].
\end{align}
Utilizing $(1-\alpha)(\beta-\beta_R) = \tilde{\beta} - \beta_R$, we can write this as
\begin{align}\label{eq:app:secondder}
  S_\alpha (\omega_{\beta_R} || \omega_\beta)'' &= \frac{2}{(1-\alpha)^3}\left[\log Z_{\beta_R} - \log Z_{\tilde{\beta}} - (\tilde{\beta}-\beta_R)E_{\tilde{\beta}} - \frac{(\tilde{\beta}-\beta_R)^2}{2}\mathrm{Var}(H)_{\tilde{\beta}}\right].
\end{align}
Here, we have introduced the average energy $E_\beta$ and the variance $\mathrm{Var}(H)_\beta = \langle H^2\rangle_\beta - \langle H\rangle_\beta^2$, which fulfill $\partial_x E_x = - \mathrm{Var}(H)_x$. With these expressions at hand, we will now show Thm. \ref{cor:simplecorollary} and another result about about the convexity of Renyi divergences for sufficiently large reference temperature $\beta$.

\subsection{Proof of Thm. \ref{cor:simplecorollary}}

We need to show that the r.h.s. of \eqref{eq:app:secondder} is positive with the premise that $x \mapsto E_x$ is convex in $x\in [ \beta_R,\beta]$, $\beta_R \leq \beta$ and $\alpha <1$. This last condition on $\alpha$ implies that we need to show that

\begin{align}
\log Z_{\beta_R} - \log Z_{\tilde{\beta}} \geq (\tilde{\beta}-\beta_R)E_{\tilde{\beta}} + \frac{(\tilde{\beta}-\beta_R)^2}{2}\mathrm{Var}(H)_{\tilde{\beta}}.
\end{align}
We use an integral representation of the l.h.s:
\begin{align}
\log Z_{\beta_R} - \log Z_{\tilde{\beta}} &= -\int^{\tilde{\beta}-\beta_R}_{0} \frac{\mathrm{d}}{\mathrm{d} x}\log Z_{\beta_R+x}\, \mathrm{d}x = \int^{\tilde{\beta}-\beta_R}_{0} E_{\beta_R+x} \mathrm{d} x.
\end{align}
Hence, we conclude that what needs to be shown is
\begin{align}\label{eq:provedinfiga}
 \int^{\tilde{\beta}-\beta_R}_{0} E_{\beta_R+x} \mathrm{d} x \geq (\tilde{\beta}-\beta_R)E_{\tilde{\beta}} + \frac{(\tilde{\beta}-\beta_R)^2}{2}\mathrm{Var}(H)_{\tilde{\beta}}.
\end{align}
Whether this inequality is satisfied, and thus, $S_{\alpha}(\rho_R\|\omega_{\beta}(H))$ is convex, is entirely determined by the function $x\mapsto E_x$.  This is due to the fact that the derivative of $E_x$ is given by $-\mathrm{Var}(H)_x$, so that the right hand side can be seen as a linear approximation to the function $E_x$. A geometrical interpretation is provided in Fig \ref{fig:illustration} showing that it is trivially satisfied when $E_x$ is convex. This finishes the proof.

As a final remark, although not useful to obtain bounds on the third law, we note that a completely analogous argument implies that if $\alpha<1$, $E_x$ is convex but in contrast to the previous case $\beta_R\geq  \beta$, then it is fulfilled that
\begin{align}\label{eq:provedinfigb}
 \int^{\tilde{\beta}-\beta_R}_{0} E_{\beta_R+x} \mathrm{d} x \leq (\tilde{\beta}-\beta_R)E_{\tilde{\beta}} + \frac{(\tilde{\beta}-\beta_R)^2}{2}\mathrm{Var}(H)_{\tilde{\beta}}.
\end{align}
This shows that in the case of resources colder than the bath, the function $S_\alpha(\rho_R\|\omega_{\beta}(H))$ is concave.

\begin{figure}
\includegraphics[width=6cm]{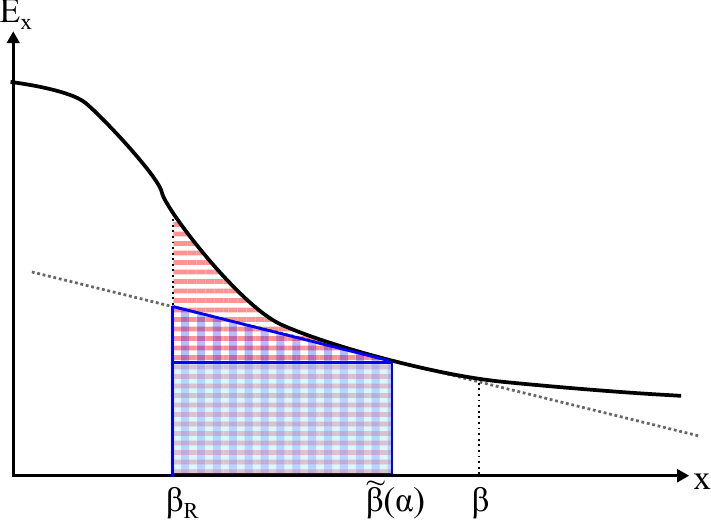}
\caption{The l.h.s. of \eqref{eq:provedinfiga} is represented by the red striped area under the curve. The r.h.s. corresponds with the blue striped region. This can be seen by noting that the blue striped region can be decomposed onto a rectangle of sides $\tilde{\beta}  - \beta_R$ and $E_{\tilde{\beta}}$ (light blue) which corresponds with the first term of the r.h.s of \eqref{eq:provedinfiga}; and a triangle that corresponds to the second term. If the function is $E_x$ convex, the red region is always larger than the blue region. }
\label{fig:illustration}
\end{figure}

\subsection{Very cold heat baths}

We will now show that in the case of very cold heat baths (very large $\beta$) we also have that $S_{\alpha}(\omega_{\beta_R}\|\omega_{\beta})$ is convex, and hence \eqref{eq:thirdlawnecessary} becomes sufficient and necessary.

\begin{theorem} For any resource of the form ($\omega_{\beta_R}(H_R),H_R)$, given a fixed $\beta_R$ there exist a sufficiently large value of $\beta$ so that \eqref{eq:thirdlawnecessary} is a sufficient and necessary condition for low temperature cooling.
\end{theorem}
\begin{proof}
We only give a sketch and show that $S_{\alpha}(\omega_{\beta_R}\|\omega_{\beta})$ is convex for values of $\alpha<\alpha_c$, where $\alpha_c<1$ is chosen arbitrarily. Recalling Eq. \ref{eq:provedinfiga}, we then need to show that
\begin{align}\label{eq:provedinfiga'}
\int^{\tilde{\beta}-\beta_R}_{0} E_{\beta_R+x} \mathrm{d} x \geq (\tilde{\beta}-\beta_R)E_{\tilde{\beta}} + \frac{(\tilde{\beta}-\beta_R)^2}{2}\mathrm{Var}(H)_{\tilde{\beta}}.
\end{align}
Note that in the limit of large $\beta$ the scaling of the r.h.s. of \eqref{eq:provedinfiga'} is such that $\tilde{\beta}-\beta_R=(1-\alpha)(\beta-\beta_R)$ scales proportionally to $\beta$, while $E_{\tilde{\beta}}$ and $\mathrm{Var}(H)_{\tilde{\beta}}$ scale as $e^{-k\beta}$. Therefore, the r.h.s. of \eqref{eq:provedinfiga'} approaches zero as $\beta \rightarrow \infty$ whereas the l.h.s. grows monotonically with $\beta$. Hence, \eqref{eq:provedinfiga'} is fulfilled which concludes the proof.
\end{proof}

\section{Equidistant levels}
Here, we consider the particular case of a system with $M+1$ equidistant levels and show that the function $E_\beta$ is convex. The energy-gap between subsequent levels is $\Delta$ and we set the ground-state energy to zero. The energy $E_\beta$ then takes the form
\begin{align}
E_\beta = \frac{1}{\e^{\beta \Delta }-1}\Delta  - \frac{M+1}{\e^{(M+1)\Delta\beta}-1}\Delta.
\end{align}
In particular, for $M\rightarrow \infty$ we obtain results for the harmonic oscillator and for $M=1$ for a qubit. We have to prove that the second derivative is positive, i.e.,
\begin{align}
E''_\beta = \frac{1}{8}\Delta^3\left[\frac{\sinh(\beta\Delta)}{\sinh(\beta\Delta/2)^4}-\underbrace{(M+1)^3\frac{\sinh((M+1)\beta\Delta)}{\sinh((M+1)\beta\Delta/2)^4}}_{=:f(\beta,M+1)}\right] \geq 0.
\end{align}
For $M=0$ this is clearly true. We will set $M+1=:\gamma$ and show that $\partial_\gamma f(\beta,\gamma) \leq 0$. We have
\begin{align}
\partial_\gamma f(\beta,\gamma) = -\gamma^2 \frac{1}{\sinh(\gamma\beta\Delta/2)^4}\left[\gamma\beta\Delta(2+\cosh(\gamma\beta\Delta))-3\sinh(\gamma\beta\Delta)\right].
\end{align}
In the following, set $\gamma\beta\Delta=x$. Due to the negative pre-factor, we are done if we can show
\begin{align}
x(2+\cosh(x)) - 3\sinh(x) \geq 0.
\end{align}
We will show this using a Taylor-expansion:
\begin{align}
2x + x \cosh(x) - 3\sinh(x) &= 2x+\sum_{n=0}^\infty x^{2n+1}\left(\frac{1}{(2n)!}-\frac{3}{(2n+1)!}\right)\\
                            &= 2x+\sum_{n=0}^\infty x^{2n+1}\left(\frac{((2n+1)-3)(2n)!}{(2n)!(2n+1)!}\right)\\
                            &= 2x - 2x + \sum_{n=1}^\infty x^{2n+1}\left(\frac{((2n+1)-3)(2n)!}{(2n)!(2n+1)!}\right)\\
&\geq 0.
\end{align}

\section{The work source model}\label{sec:worksource}

Here we show that a work source of the form $(\rho_R^w, H_R^w)$ as given in the main text fulfills the second law of thermodynamics. Let us consider an arbitrary system $(\rho_S,H_S)$ and let us consider catalytic thermal operations on $SW$. We will show that the maximum amount of mean energy that one can store on the work source is bounded by the initial non-equilibrium free energy of $S$. Let us recall, see for instance Ref. \cite{Skrzypczyk2011}, that the free energy difference is given by
\begin{equation}
\Delta F_{\beta} (\rho,H):=1/\beta S(\rho\| \omega_{\beta}(H)) =F(\rho,H) - F(\omega_{\beta}(H)),
\end{equation}
where $F(\rho,H)=\tr (\rho H) - \frac{1}{\beta} S(\rho)$ is the free energy.

The protocol of work extraction is a transition of the form
 \begin{equation}
 \rho_{SW}^i:=\rho_S \otimes \rho_R^w \rightarrow \rho_{SW}^f= \mc{E}(\rho_{SW}^i)
 \end{equation}
 where $\mc{E}$ is any channel that has the Gibbs state as a fixed point. Monotonicity of $\Delta F_{\beta}$ under channels of the form $\mc{E}$ imply that
\begin{equation}
\Delta F^f_{\beta} := \Delta F_{\beta}(\rho_SR^f,H_S+ H_R) \leq \Delta F_{\beta}(\rho_{SR}^i,H_S+ H_R) := \Delta F_{\beta}^i.
\end{equation}
Combining this last equation with super-additivity and additivity of the relative entropy one can easily find, following a similar reasoning as in Ref. \cite{Skrzypczyk2011}, that
\begin{equation}
\Delta E_R \leq \Delta F_{\beta} (\rho_S,H_S)
\end{equation}
where $\Delta E_R=\tr (\id \otimes H_R^w\:\: \rho_{SR}^f)- \tr (\id \otimes H_R^w  \:\: \rho_{SR}^i)$ is the mean energy stored in the work source $R$.

\section{Arbitrary target-states close to the ground-state}\label{sec:app:epsilon}
In this section we prove a result similar to our general sufficient condition for cooling, but where we consider target states of the form
\begin{align}
\rho_\epsilon = (1-\epsilon) \ketbra{0}{0} + \epsilon \rho^\perp,\quad \epsilon \ll 1.
\end{align}
where $\rho^\perp$ is a density matrix which has full rank on the subspace orthogonal to the ground-state $\ket{0}$ and commutes with the Hamiltonian $H_S$.

\begin{theorem}[General sufficient condition for cooling] \label{thm:generaltheorem_epsilon} For every choice of $\beta$, $H_S$ and $\rho^\perp$ as above, there is a critical $\epsilon_{\mathrm{cr}}>0$ such that for any $\epsilon < \epsilon_{\rm cr}$ the condition
\begin{align}\label{eq:main_result_epsilon}
\mc V_\beta(\rho_R,H_R) +\tilde{K}(\epsilon,\beta,\rho_R,H_R,H_S,\rho^\perp) \geq \mc V_\beta(\rho_\epsilon,H_S)
\end{align}
is \emph{sufficient} for cooling. The function $\tilde{K}$ has the property $\tilde{K}(\epsilon,\beta,\rho_R,H_R,\rho^\perp) \rightarrow 0$ as $\epsilon \rightarrow \infty$ for any fixed $\beta,H_R,\rho_R,\rho^\perp$ and $H_S$.
\end{theorem}

The proof of this theorem is essentially identical to the one of Theorem~\ref{thm:generaltheorem}. The only difference is that instead of Lemma~\ref{lemma:concavity}, we use the following Lemma:
\begin{lemma}[Concavity close to ground-state]\label{lemma:concavity_epsilon}
Let $H_S$ be a $d$-dimensional Hamiltonian with ground-state $\ket{0}$ and $H\ket{0}=0$. Let $\beta>0$ be fixed and consider the state
\begin{align}
\rho_\epsilon = (1-\epsilon) \ketbra{0}{0} + \epsilon \rho^\perp,
\end{align}
with $\mathrm{rank}(\rho^\perp)=d-1$, $\rho^\perp\ket{0}=0$ and $[\rho^\perp,H_S]=0$. Then there exists an $\epsilon_{\mathrm{cr}}>0$ such that for all $\alpha<\delta(\epsilon)$
\begin{align}
\frac{\mathrm{d}^2}{\mathrm{d}\alpha^2 }S_\alpha(\rho_\epsilon||\omega_\beta(H_S)) <0, \quad \forall \epsilon<\epsilon_{\mathrm{cr}}.
\end{align}
Here, $\delta(\epsilon)$ fulfills
\begin{align}
\delta(\epsilon) = \frac{\log Z_\beta}{\vac_\beta(\rho_\epsilon,H_S)} < 1,\quad \forall \epsilon<\epsilon_{\mathrm{cr}}.
\end{align}
\end{lemma}

We will now proof this Lemma. Let us then express the Renyi-divergence as
\begin{align}
S_\alpha(\rho_\epsilon||\omega_\beta(H_S)) &= \frac{1}{\alpha-1}\log\left((1-\epsilon)^\alpha + \epsilon^\alpha\tr((\rho^\perp)^\alpha \e^{-\beta H_S(1-\alpha)})\right) + \log(Z_S) \\
&=: \frac{1}{\alpha-1}\log (f(\alpha)) + \log(Z_S).
\end{align}
As is apparent from the expression, in the following we will often encounter the functions
\begin{align}
\tilde{f}(\alpha) &:= \tr((\rho^\perp)^\alpha \e^{-\beta H(1-\alpha)}),\\
f_\epsilon(\alpha) &:= \tr(\rho_\epsilon^\alpha \e^{-\beta H(1-\alpha)}) = (1-\epsilon)^\alpha + \epsilon^\alpha\tilde{f}(\alpha).
\end{align}
It is useful to remember from the main text that $\rho^\perp$ is a normalized quantum state that commutes with $H$ and has rank $d-1$.
In the following we will also often write $S_\alpha$ instead of $S_\alpha(\rho_\epsilon||\omega_\beta(H_S))$ and simply $f_\epsilon$ or $f_{\epsilon,\alpha}$ instead of $f_\epsilon(\alpha)$ to simplify the notation (similarly for $\tilde{f}$). While $f_\epsilon$ and $\tilde{f}$ are structurally essentially the same, it is important to keep in mind that only $f_\epsilon$, and not $\tilde{f}$, depends on $\epsilon$.

We now have to prove that $S_\alpha$ is concave for small enough $\epsilon$, i.e., have to show that there exists a $\epsilon_{\mathrm{cr}}>0$ such that  its second derivative is negative for $\epsilon < \epsilon_{\mathrm{cr}}$.
The second derivative of $S_\alpha$ can be computed straightforwardly and gives
\begin{align}\label{eq:secondderivative}
S''_\alpha &= -\frac{2}{(1-\alpha)^3}\log f_\epsilon - \frac{2}{(1-\alpha)^2}\frac{f_\epsilon'}{f_\epsilon} + \frac{1}{1-\alpha}\left(\left(\frac{f_\epsilon'}{f_\epsilon}\right)^2-\frac{f_\epsilon''}{f_\epsilon}\right).
\end{align}
To go on we need to establish a few properties of functions like $f_\epsilon$ and $\tilde{f}$. We will collect these properties in a series of Lemmata.

\begin{lemma}\label{lemma:lemma1}
Let $\rho$ be a quantum state and $\sigma$ be a positive semi-definite operator with $[\rho,\sigma]=0$. Define $f(\alpha):=\tr(\rho^{\alpha}\sigma^{1-\alpha})$.
Then
\begin{align}
 (f')^2 - f'' f \leq 0,\quad 0<\alpha<1.
\end{align}
\begin{proof}
A simple calculation shows that
\begin{align}
f'(\alpha) &= \tr(\rho^\alpha (\log(\rho) - \log(\sigma)) \sigma^{1-\alpha}),\\
f''(\alpha) &= \tr(\rho^\alpha (\log(\rho) - \log(\sigma))^2 \sigma^{1-\alpha}).
\end{align}
We now use the Cauchy-Schwartz inequality $|\tr(A^\dagger B\rho)|^2 \leq \tr(A^\dagger A \rho)\tr(B^\dagger B\rho)$ with $A=\rho^{-\alpha/2} (\log(\rho) - \log(\sigma)) \sigma^{\alpha/2}$ and $B = \rho^{-\alpha/2}\sigma^{\alpha/2}$ to obtain (note the change from $\alpha$ to $1-\alpha$)
\begin{align}
f'(1-\alpha)^2 &= \tr\left(\rho^{-\alpha/2} (\log(\rho) - \log(\sigma)) \sigma^{\alpha/2} \rho^{-\alpha/2}\sigma^{\alpha/2} \rho\right)^2\nonumber \\
&\leq \tr\left(\rho^{-\alpha} (\log(\rho) - \log(\sigma))^2 \sigma^{\alpha}\rho \right)\tr(\rho^{-\alpha}\sigma^{\alpha}\rho) \nonumber \\
&= \tr\left(\rho^{1-\alpha} (\log(\rho) - \log(\sigma))^2 \sigma^{\alpha} \right)\tr\left(\rho^{1-\alpha}\sigma^\alpha\right) \nonumber \\
&= f''(1-\alpha) f(1-\alpha).
\end{align}
\end{proof}
\end{lemma}

\begin{lemma}\label{lemma:upperboundf}
Let $H$ be a Hamiltonian with ground-state energy $E_0=0$ and let $sigma$ be a quantum state with $[\sigma,H]=0$. Then
\begin{align}
f(\alpha) := \tr(\sigma^\alpha \e^{-\beta H(1-\alpha)}) \leq Z,\quad 0\leq \alpha \leq 1.
\end{align}
\begin{proof}
From the calculation of the previous Lemma we see that the second derivative of $f$ is the trace of a product of positive commuting operators. Hence it is always positive and therefore $f$ is convex. But since $H\geq 0$ we have $f(0) = Z \geq 1 = f(1)$ and from convexity we get $f(\alpha) \leq Z$.
\end{proof}
\end{lemma}

Note that due to our assumption about the groundstate energy we have $Z_S\geq 1$ and from the above Lemma we know $f_\epsilon \leq Z_S$. We will now show that for every $0 < \alpha_c' < 1$, we have $1\leq f_\epsilon(\alpha)\leq Z_S$ if $\epsilon$ is small enough and $\alpha < \alpha_c'$.

\begin{lemma}\label{lemma:lowerboundf}
For any $0< \alpha_c'< 1$ there exists a critical $\epsilon'_{\mathrm{cr}}>0$, such that for all $\epsilon < \epsilon'_{\mathrm{cr}}$ we have
\begin{align}
f_\epsilon(\alpha) &:= \tr(\rho_\epsilon^\alpha \e^{-\beta H(1-\alpha)}) \geq 1, \quad 0\leq \alpha < \alpha_c'.
\end{align}
\begin{proof}
Assume some $0 < \alpha < \alpha_c'$. Using that $\tilde{f}$ is independent of $\epsilon$ and positive, we can lower bound it by some $\tilde{f}_{\mathrm{min}}>0$. Also, $(1-\alpha)^\alpha$ is monotonically decreasing with $\alpha$ for $0<\alpha < 1$. We therefore get the lower bound
\begin{align}
f_\epsilon(\alpha) = (1-\epsilon)^\alpha + \epsilon^\alpha \tilde{f}(\alpha) \geq (1-\epsilon) + \epsilon^\alpha\tilde{f}_{\mathrm{min}} = 1 + \epsilon(\epsilon^{\alpha-1} \tilde{f}_{\mathrm{min}}-1).
\end{align}
Thus for $\epsilon < \epsilon'_{\mathrm{cr}} (\alpha_c') := \left(\tilde{f}_\mathrm{min}\right)^{\frac{1}{1-\alpha_c}}$ we have $f_\epsilon(\alpha) \geq 1$.
\end{proof}
\end{lemma}

\begin{remark}Due to the preceding Lemma, we will in the following take the (somewhat arbitrary) choice $\alpha_c' = 1/3$ and only consider $\alpha < \alpha_c'$ as well as values of $\epsilon < \epsilon_{\mathrm{cr}}'(\alpha_c')$. Since later we are anyway only interested in arbitrarily small values of $\epsilon$ and $\alpha \leq \delta(\epsilon)$, this is no obstruction.
\end{remark}

\begin{lemma}\label{lemma:firstderivativef}
For all $\alpha < \alpha_c'$ and $\epsilon < \epsilon_{\mathrm{cr}}'$ we have $f_\epsilon'(\alpha) \leq 0$.
\begin{proof}
Follows from $f_\epsilon(\alpha) \geq 1$ for all $\alpha \leq \alpha_c'$ together with the facts that $f_\epsilon(1)=1$,$f_\epsilon(0)=Z_S$ and that $f_\epsilon$ is convex.
\end{proof}
\end{lemma}

We are now in position to go on with the proof of the asymptotic concavity. First, we will further restrict the values of $\alpha$ by choosing arbitrarily $\alpha_c < \alpha_c'=1/3$ and restricting to $\alpha \leq \alpha_c$. The reason for this will become clear later in the proof.

Considering eq.~\eqref{eq:secondderivative} and using $Z_S \geq f_\epsilon \geq 1$ as well as Lemma~\ref{lemma:lemma1}, we can upper now bound the second derivative as
\begin{align}\label{eq:secondderivative2}
  S''_\alpha &\leq - \frac{2}{(1-\alpha)^2}\frac{f_\epsilon'}{f_\epsilon} + \frac{1}{(1-\alpha)f_\epsilon^2}\left(\left(f_\epsilon'\right)^2-f_\epsilon''f_\epsilon\right) \\
            &\leq - \frac{2}{(1-\alpha_c)^2}\frac{f_\epsilon'}{f_\epsilon} + \frac{1}{Z_S^2}\left(\left(f_\epsilon'\right)^2-f_\epsilon''f_\epsilon\right).
\end{align}
One might be tempted to use Lemma~\ref{lemma:lemma1} and simply upper bound the second term by zero, but that bound would be too weak, since the first term diverges as $\log(1/\epsilon)$. We will therefore now have to do a more detailed calculation. We first compute the derivatives of $f_\epsilon$:
\begin{align}
f_{\epsilon,\alpha}' &= (1-\epsilon)^\alpha\log(1-\epsilon) + \log(\epsilon)\epsilon^\alpha\tilde{f}_\alpha + \epsilon^\alpha \tilde{f}'_\alpha,\\
(f_{\epsilon,\alpha}')^2 &= (1-\epsilon)^{2\alpha}\log(1-\epsilon)^2 + \log(\epsilon)^2\epsilon^{2\alpha}\tilde{f}^2_\alpha + \epsilon^{2\alpha} (\tilde{f}'_\alpha)^2 + 2(1-\epsilon)^\alpha\log(1-\epsilon)\log(\epsilon)\epsilon^\alpha \tilde{f}_\alpha\nonumber \\
&\quad + 2(1-\epsilon)^\alpha\log(1-\epsilon)\epsilon^\alpha \tilde{f}'_\alpha + 2 \log(\epsilon)\epsilon^{2\alpha}\tilde{f}_\alpha \tilde{f}'_\alpha.\\
f_{\epsilon,\alpha}'' &= (1-\epsilon)^\alpha\log(1-\epsilon)^2 + \log(\epsilon)^2\epsilon^\alpha\tilde{f}_\alpha+2\log(\epsilon)\epsilon^\alpha\tilde{f}'_\alpha + \epsilon^\alpha \tilde{f}''_\alpha.
\end{align}
These give
\begin{align}
(f_{\epsilon,\alpha}')^2 - f_{\epsilon,\alpha}'' \epsilon^\alpha \tilde{f}_\alpha &= (1-\epsilon)^{2\alpha}\log(1-\epsilon)^2 + \log(\epsilon)^2\epsilon^{2\alpha}\tilde{f}^2_\alpha + \epsilon^{2\alpha} (\tilde{f}'_\alpha)^2 + 2(1-\epsilon)^\alpha\log(1-\epsilon)\log(\epsilon)\epsilon^\alpha \tilde{f}_\alpha \nonumber\\
&\quad + 2(1-\epsilon)^\alpha\log(1-\epsilon)\epsilon^\alpha \tilde{f}'_\alpha + 2 \log(\epsilon)\epsilon^{2\alpha}\tilde{f}_\alpha \tilde{f}'_\alpha \nonumber \\
&\quad - \epsilon^\alpha \tilde{f}_\alpha\left((1-\epsilon)^\alpha\log(1-\epsilon)^2 + \log(\epsilon)^2\epsilon^\alpha\tilde{f}_\alpha+2\log(\epsilon)\epsilon^\alpha\tilde{f}'_\alpha + \epsilon^\alpha \tilde{f}''_\alpha\right) \\
&= \log(1-\epsilon)^2 (1-\epsilon)^\alpha\left((1-\epsilon)^\alpha -\epsilon^\alpha \tilde{f}_\alpha\right)  + \epsilon^{2\alpha} (\tilde{f}'_\alpha)^2 \nonumber \\
&\quad + 2(1-\epsilon)^\alpha\log(1-\epsilon)\log(\epsilon)\epsilon^\alpha \tilde{f}_\alpha + 2(1-\epsilon)^\alpha\log(1-\epsilon)\epsilon^\alpha \tilde{f}'_\alpha - \epsilon^{2\alpha} \tilde{f}''_\alpha \tilde{f}_\alpha\\
&= \log(1-\epsilon)^2 (1-\epsilon)^\alpha\left((1-\epsilon)^\alpha -\epsilon^\alpha \tilde{f}_\alpha\right)  + 2(1-\epsilon)^\alpha\log(1-\epsilon)\epsilon^\alpha\left(\log(\epsilon) \tilde{f}_\alpha + \epsilon^\alpha \tilde{f}'_\alpha\right) \nonumber \\
&\quad +  \epsilon^{2\alpha}\left((\tilde{f}'_\alpha)^2- \tilde{f}''_\alpha \tilde{f}_\alpha\right).
\end{align}
Hence we have
\begin{align}
(f_{\epsilon,\alpha}')^2 - f_{\epsilon,\alpha}'' f_\alpha &= -(1-\epsilon)^\alpha f_{\epsilon,\alpha}'' + \log(1-\epsilon)^2 (1-\epsilon)^\alpha\left((1-\epsilon)^\alpha -\epsilon^\alpha \tilde{f}_\alpha\right) \nonumber \\
&\quad + 2(1-\epsilon)^\alpha\log(1-\epsilon)\epsilon^\alpha\left(\log(\epsilon) \tilde{f}_\alpha + \epsilon^\alpha \tilde{f}'_\alpha\right) +  \epsilon^{2\alpha}\left((\tilde{f}'_\alpha)^2- \tilde{f}''_\alpha \tilde{f}_\alpha\right).
\end{align}
Note in particular that the last term is negative semi-definite due to Lemma~\ref{lemma:lemma1}.

Let us now also write $0< 1/k_c := (1-\alpha_c)^2 < 1$. Inserting the previous result into eq.~\eqref{eq:secondderivative2} we then obtain
\begin{align}
S''_\alpha &\leq - 2 k_c f_{\epsilon,\alpha}' + \frac{1}{Z^2}\left(\left(f_{\epsilon,\alpha}'\right)^2-f_{\epsilon,\alpha}'' f_\alpha\right) \\
&\leq -2 k_c f_{\epsilon,\alpha}' + \frac{1}{Z^2}
\left( -(1-\epsilon)^\alpha f_{\epsilon,\alpha}''
  + \log(1-\epsilon)^2 (1-\epsilon)^\alpha\left((1-\epsilon)^\alpha -\epsilon^\alpha \tilde{f}_\alpha\right) \right.\nonumber \\
  &\quad\quad\quad\quad\quad\quad\quad\left.+ 2(1-\epsilon)^\alpha\log(1-\epsilon)\epsilon^\alpha\left(\log(\epsilon) \tilde{f}_\alpha
  + \epsilon^\alpha \tilde{f}'_\alpha\right) \right)\\
&\leq -2 k_c f_{\epsilon,\alpha}' + \frac{1}{Z^2}
\left( -(1-\epsilon) f_{\epsilon,\alpha}''
  + \log(1-\epsilon)^2 \left(1 -(1-\epsilon)\epsilon \tilde{f}_\alpha\right) \right.\nonumber \\
  &\quad\quad\quad\quad\quad\quad\quad\left.+ 2\log(1-\epsilon)\left(\log(\epsilon) \tilde{f}_\alpha
  + \epsilon^{2\alpha}(1-\epsilon)^\alpha \tilde{f}'_\alpha\right) \right)
\end{align}
We now lower bound $f_\epsilon'$ and $f_\epsilon''$ as
\begin{align}
f_{\epsilon,\alpha}' &\geq \log(1-\epsilon) + \log(\epsilon)\tilde{f}_\alpha + \tilde{f}_\alpha'\\
f_{\epsilon,\alpha}'' &\geq (1-\epsilon)\log(1-\epsilon)^2 + \log(\epsilon)^2 \epsilon^{\alpha_c} \tilde{f}_\alpha + 2\log(\epsilon)\epsilon^\alpha\tilde{f}'_\alpha + \epsilon \tilde{f}''_\alpha.
\end{align}
Note that we cannot easily bound the terms involving $\tilde{f}'_\alpha$ since we do not know the sign of $\tilde{f}'_\alpha$.  However, we emphasize again that $\tilde{f}$ is independent of $\epsilon$ and can hence essentially be treated as constant. Putting in the bounds then yields
\begin{align}
S''_\alpha &\leq -2k_c (\log(1-\epsilon) + \log(\epsilon)\tilde{f}_\alpha + \tilde{f}_\alpha') -\frac{1-\epsilon}{Z^2}\left((1-\epsilon)\log(1-\epsilon)^2 + \log(\epsilon)^2 \epsilon^{\alpha_c} \tilde{f}_\alpha + 2\log(\epsilon)\epsilon^\alpha\tilde{f}'_\alpha + \epsilon \tilde{f}''_\alpha\right) \nonumber \\ &\quad+ \frac{1}{Z^2}
\left( \log(1-\epsilon)^2 \left(1 -(1-\epsilon)\epsilon \tilde{f}_\alpha\right) + 2\log(1-\epsilon)\left(\log(\epsilon) \tilde{f}_\alpha
  + \epsilon^{2\alpha}(1-\epsilon)^\alpha \tilde{f}'_\alpha\right) \right)\\
&= \log(\epsilon)\tilde{f}_\alpha \left(\frac{2}{Z^2}\log(1-\epsilon)-\frac{1-\epsilon}{Z^2}2\epsilon^\alpha \frac{\tilde{f}'_\alpha}{\tilde{f}_\alpha} -\frac{1-\epsilon}{Z^2}\log(\epsilon)\epsilon^{\alpha_c} - \frac{2}{(1-\alpha_c)^2} \right) \nonumber \\
&\quad + \log(1-\epsilon)\left( - \frac{2}{(1-\alpha_c)^2} - \frac{(1-\epsilon)^2}{Z^2}\log(1-\epsilon) + \frac{1-(1-\epsilon)\epsilon\tilde{f}_\alpha}{Z^2}\log(1-\epsilon) + \frac{2}{Z^2}\epsilon^{2\alpha}(1-\epsilon)^\alpha\tilde{f}'_\alpha \right)\nonumber  \\
&\quad -\frac{2}{(1-\alpha_c)^2}\tilde{f}'_\alpha - \frac{(1-\epsilon)\epsilon}{Z^2}\tilde{f}''_\alpha \\
&\leq \log(\epsilon)\tilde{f}_\alpha \left(\frac{2}{Z^2}\log(1-\epsilon)-\frac{1-\epsilon}{Z^2}2\epsilon^\alpha \frac{\tilde{f}'_\alpha}{\tilde{f}_\alpha} -\frac{1-\epsilon}{Z^2}\log(\epsilon)\epsilon^{\alpha_c} - \frac{2}{(1-\alpha_c)^2} \right) + M(\epsilon, H,\beta,\rho^\perp) - K(\alpha_c,H,\beta, \rho^\perp),
\end{align}
where $M$ goes to zero as $\epsilon$ goes to zero and $K$ is independent of $\epsilon$. Also note that $M$ is bounded and independent of $\alpha$ (due to the boundedness of $\tilde{f}_\alpha$ and its derivatives).
Let us define $m(\alpha_c) = \max_{\alpha\leq \alpha_c} \tilde{f}'_\alpha/ \tilde{f}_\alpha$. Since $\alpha_c < 1/2$ we can simplify the bound to
\begin{align}
S''_\alpha &\leq \log(\epsilon)\frac{\tilde{f}_\alpha}{Z^2} \left(2 \log(1-\epsilon)-2(1-\epsilon)\epsilon^\alpha m(\alpha_c) -(1-\epsilon)\log(\epsilon)\epsilon^{\alpha_c} - 8 Z^2 \right) + M(\epsilon, H,\beta,\rho^\perp) - K(\alpha_c,H,\beta, \rho^\perp).
\end{align}
Clearly $S''_\alpha$ can be made negative by taking $\epsilon$ and $\alpha_c$ arbitrarily small since the dominant term in the bracket goes as $-\log(\epsilon)$. However, since our objective is to upper bound $S_\alpha$ by $\mc V_\beta(\rho_\epsilon,H_S) \alpha$ for all $\alpha\leq \alpha_c$, we also need that $\mc V_\beta(\rho_\epsilon,H_S) \alpha_c \geq \log Z_S$ and hence $\alpha_c \geq \log(Z_S)/\mc V_\beta(\rho_\epsilon,H_S)$. Hence we choose $\alpha_c = \delta(\epsilon) = \log Z_S/\mc V_\beta(\rho_\epsilon,H_S)$ and hope for the best. The vacancy is given by:
\begin{align}
\mc V_\beta(\rho_\epsilon,H_S) &= - \log(1-\epsilon)\frac{1}{Z_S} + \frac{Z_S-1}{Z_S}\log(1/\epsilon) + C_1(\rho^\perp,\beta, H_S),
\end{align}
where $C_1$ does not depend on $\epsilon$. Hence
\begin{align}
\lim_{\epsilon\rightarrow 0} \epsilon^{\delta(\epsilon)} = \frac{1}{Z_S^{\frac{Z_S}{Z_S-1}}}\quad\text{and}\quad
\lim_{\epsilon\rightarrow 0 } - (1-\epsilon)\log(\epsilon)\epsilon^{\delta(\epsilon)} = +\infty.
\end{align}
Since the other terms in the first bracket in $S''_\alpha$ go to zero as $\epsilon\rightarrow 0$ and $K$ is independent of $\epsilon$, we can thus find a finite $\epsilon_{cr}$ such that
\begin{align}
S''_\alpha \leq 0,\quad \alpha\leq \delta(\epsilon_{cr}).
\end{align}
This finishes the proof.\qed

\he{
\section{Exactly conserved catalysts}
\label{app:sec:exact_catalysts}
In this section we analyze the scenario where \ro{the catalyst always has to be} returned without any error. \ro{That is, the cooling protocol considers a process like the one described in Sec. \ref{sec:set-up} but taking $\epsilon=0$.} First, note that the vacancy is automatically also a monotone in this setting, since we are considering a subset of free operations. Hence, the inequality
\begin{align}
V_\beta(\rho_R,H_R) \geq V_\beta(\rho_S,H_S)
\end{align}
\ro{is also \he{a necessary condition} for this set of free operations}.
In the following, we will consider for simplicity only the case where the target system is thermal, $\rho_S = \omega_{\beta_S}(H_S)$.

We will now prove the following theorem, which provides a sufficient condition for cooling and which coincides with that in our general Theorem~\ref{thm:generaltheorem} up to a multiplicative factor.

\begin{theorem}[Sufficient condition under exact catalysis] \label{thm:generaltheorem_exact_catalysts} Assume thermal operations with exact catalysts. Then for every choice of $\beta$ and $H_S$ there is a critical $\beta_{\mathrm{cr}}>0$ such that for any $\beta_S > \beta_{\rm cr}$ and full-rank resource $(\rho_R,H_R)$ (diagonal in the energy eigenbasis) the condition
\begin{align}\label{eq:main_result}
\mc V_\beta(\rho_R,H_R) - K(\beta_S,\beta,\rho_R,H_R,H_S) \geq r(\beta,H)\mc V_\beta(\omega_{\beta_S}(H_S),H_S)
\end{align}
is \emph{sufficient} for cooling. The positive semidefinite function $K$ is identical to that in Theorem~\ref{thm:generaltheorem} and the constant $r(\beta,H_S)$ is \ro{independent of $\rho_R,H_R$ and $\beta_S$ and given by
\begin{equation}
r(\beta,H_S)= 1+2\frac{E_{\mathrm{max}}-E_\beta}{E_\beta}
\end{equation}
where $E_{\mathrm{max}}$ is the largest eigenvalue of $H_S$ \he{and we assume that the ground-state energy of $H_S$ is zero.}}
\end{theorem}
\ro{Before going into the proof, let us discuss briefly the implications that the correction given by $r(\beta,H_S)$ has over the sufficient condition Thm. \ref{thm:generaltheorem}. This is better explained if we look at the scaling results of Sec. \ref{sec:iid}. There we showed that the sufficient condition of Thm. \ref{thm:generaltheorem} provides an upper bound on the number of copies of a resource that are sufficient to implement a cooling process, as given by $n^{\text{suff}}$ in \eqref{eq:iidupperbound}. The sufficient condition laid out in Thm. \ref{thm:generaltheorem_exact_catalysts} implies simply that that we need $r$-times more systems to implement the cooling protocol, where $r=r(\beta,H)$. Note importantly that $r$ does not depend on the final temperature, so employing $r(\beta,H_S) \times n^{\text{suff}}$ is always sufficient for cooling.} \he{We emphasize that we believe that the factor $r(\beta,H_S)$ can be made much closer to $1$ by more elaborate proof techniques, but leave this as an open problem. }

\begin{proof}
It was shown in Ref.~\cite{Brandao2015} that a transition $\rho\rightarrow \rho'$ between two diagonal states is possible with exact preservation of the catalyst if and only if the Renyi-divergences
\begin{align}
S_\alpha(\rho \| \omega_\beta(H)) = \frac{\mathrm{sign}(\alpha)}{\alpha-1} \log \tr\left(\rho^\alpha \| \omega_\beta(H)^{1-\alpha}\right)
\end{align}
do not increase for \emph{all} $\alpha \in (-\infty,+\infty)$. The sufficient condition in Theorem~\ref{thm:generaltheorem} covers all $\alpha\geq 0$. We thus have to check that we can fulfill all the inequalities for $\alpha<0$ using the multiplicative factor $r(\beta,H_S)$. To do this we provide new lower- and upper-bounds for the Renyi-divergences for negative $\alpha$. We begin with a lower-bound. Consider any state $\rho$ with eigenvalues $p_i$ in the energy-eigenbasis. Then we have
\begin{align}
S_{-|\alpha|}(\rho \| \omega_\beta(H)) = \frac{1}{|\alpha|+1} \log \left(\sum_i p_i^{-|\alpha|} w_i^{1+|\alpha|}\right),
\end{align}
where $w_i = \e^{-\beta E_i}/Z_\beta$ are the eigenvalues of the thermal state. Using concavity of the logarithm we can bound this as
\begin{align}\label{eq:upper_bound_renyi}
S_{-|\alpha|}(\rho \| \omega_\beta(H)) &\geq \frac{1}{|\alpha|+1} \sum_i w_i \log\left( p_i^{-|\alpha|} w_i^{|\alpha|}\right) = \frac{|\alpha|}{|\alpha|+1} \sum_i\left(w_i \log(w_i)- w_i \log(p_i) \right)\\
&= \frac{|\alpha|}{|\alpha|+1} V_\beta(\rho,H).
\end{align}
We can thus lower bound all the Renyi-divergences for negative $\alpha$ by a simple function. Later, we will apply this bound to the resource.

We will now derive a similar \emph{upper bound} for the target system, i.e., assuming a system in a thermal state.
First, we rewrite the Renyi-divergences as
\begin{align}
S_{-|\alpha|}(\omega_{\beta_S}(H_S) \| \omega_\beta(H_S)) = \frac{|\alpha|}{|\alpha|+1}\log(Z_{\beta_S}) - \log(Z_\beta) + \frac{1}{1+|\alpha|} \log \tr\left(\e^{(\beta_S-\beta)|\alpha|H_S}\e^{-\beta H_S}\right),
\end{align}
which can be verified by direct calculation. We will now use the \emph{log-sum inequality}. It states that for any two sets of $d$ non-negative numbers $\{a_i\}$ and $\{b_i\}$ we have
\begin{align}
 \log \frac{a}{b} \leq \sum_i \frac{a_i}{a} \log \frac{a_i}{b_i},
\end{align}
with $a=\sum_i a_i$ and $b=\sum_i b_i$. Let $E_i$ be the energy-eigenvalues of $H_S$. Then we set
\begin{align}
a_i = \e^{(\beta_S-\beta)|\alpha| E_i} \e^{-\beta E_i} = \e^{-\ro{\tilde{\beta}(\alpha)E_i}},\quad b_i = \frac{\e^{-\beta E_i}}{Z_\beta},
\end{align}
\ro{where $\tilde{\beta}(\alpha):= \beta - (\beta_S -\beta)|\alpha|$}. Using the log-sum inequality we then obtain
\begin{align}
S_{-|\alpha|}(\omega_{\beta_S}(H_S) \| \omega_\beta(H_S)) &\leq \frac{|\alpha|}{|\alpha|+1}\log(Z_{\beta_S}) - \log(Z_\beta) + \frac{1}{1+|\alpha|} \sum_i \frac{\e^{-\tilde{\beta}(\alpha)E_i}}{Z_{\tilde{\beta}(\alpha)}} \log\left( \e^{(\beta_S-\beta)|\alpha| E_i} Z_\beta \right)\\
&= \frac{|\alpha|}{|\alpha|+1}\left(\log(Z_{\beta_S}) - \log(Z_\beta)\right)  + \frac{|\alpha|}{|\alpha|+1}(\beta_S-\beta) \sum_i \frac{\e^{-\tilde{\beta}(\alpha)E_i}}{Z_{\tilde{\beta}(\alpha)}} E_i.
\end{align}
Denoting by $E_{\mathrm{max}}$ the maximum energy, we then get the bound
\begin{align}
S_{-|\alpha|}(\omega_{\beta_S}(H_S) \| \omega_\beta(H_S)) &\leq\frac{|\alpha|}{|\alpha|+1} \left(\log(Z_{\beta_S}) - \log(Z_\beta) + (\beta_S-\beta) E_{\mathrm{max}}\right).
\end{align}
\ro{Let us recall from Eq. \eqref{eq:vac_as_F}, that for thermal states, the vacancy can be expressed as a function of the non-equilibrium free energy as}
\begin{align}
 \mc V_\beta(\omega_{\beta_S}(H_S),H_S) = \beta_S \Delta F_{\beta_S}(\omega_\beta(H),H) = \beta_S E_\beta - S_\beta + \log(Z_{\beta_S}),
\end{align}
where $E_\beta$ and $S_\beta$ denote the thermal energy expectation value and von~Neumann entropy at inverse temperature $\beta$. Using this together with $-\log{Z_\beta} = \beta E_\beta - S_\beta$, we can rewrite the upper bound on the Renyi-divergences as
\begin{align}
S_{-|\alpha|}(\omega_{\beta_S}(H_S) \| \omega_\beta(H_S)) &\leq
\frac{|\alpha|}{|\alpha|+1} \left( \mc V_\beta(\omega_{\beta_S}(H_S),H_S)\left[1+\frac{E_{\mathrm{max}}-E_\beta}{\Delta F_{\beta_S}(\omega_\beta(H),H)} \right] - \beta(E_{\mathrm{max}}-E_\beta)\right) \\
&\leq \frac{|\alpha|}{|\alpha|+1}  \mc V_\beta(\omega_{\beta_S}(H_S),H_S)\left[1+\frac{E_{\mathrm{max}}-E_\beta}{\Delta F_{\beta_S}(\omega_\beta(H),H)} \right].
\end{align}
Since we have
\begin{align}
\Delta F_{\beta_S}(\omega_\beta(H),H) = E_\beta - E_{\beta_S} - \frac{1}{\beta_S}(S_\beta - S_{\beta_S}),
\end{align}
\ro{it is always possible find a critical inverse temperature $\beta'_S$ such that $\Delta F_{\beta_S}(\omega_\beta(H),H)\geq E_\beta/2$ for all $\beta_S>\beta'_S$.} \ro{Then, for} all $\beta_S$ larger than this critical temperature we can bound the Renyi-divergences as
\begin{align}
S_{-|\alpha|}(\omega_{\beta_S}(H_S) \| \omega_\beta(H_S)) &\leq \frac{|\alpha|}{|\alpha|+1}  \mc V_\beta(\omega_{\beta_S}(H_S),H_S)\left[1+2\frac{E_{\mathrm{max}}-E_\beta}{E_\beta} \right] =: \frac{|\alpha|}{|\alpha|+1}  \mc V_\beta(\omega_{\beta_S}(H_S),H_S) r(\beta,H_S).
\end{align}
Using the lower bound \eqref{eq:upper_bound_renyi} for the resource, we then find that the inequalities for negative $\alpha$ are fulfilled if
\begin{align}
\frac{|\alpha|}{|\alpha|+1} \mc V_\beta(\rho_R,H_R) \geq \frac{|\alpha|}{|\alpha|+1} \mc V_\beta(\omega_{\beta_S}(H_S),H_S) r(\beta,H_S)
\end{align}
Cancelling the prefactors, we get as sufficient condition for negative values of $\alpha$ the inequality
\begin{align}
\mc V_\beta(\rho_R,H_R) \geq \mc V_\beta(\omega_{\beta_S}(H_S),H_S) r(\beta,H_S).
\end{align}
Combining this with the sufficient condition for positive $\alpha$, \ro{which is the sufficient condition provided by Thm. \ref{thm:generaltheorem},} then yields the claimed sufficient condition in the theorem.
\end{proof}

\subsubsection{Catalysts can always be chosen with full rank}
Before finishing this section, let us point out that in the case of exact catalysis, one can always choose the catalyst to have full rank. \ro{That is, the actual implementation of the cooling protocol, which is guaranteed to exist under the conditions of Thm. \ref{thm:generaltheorem_exact_catalysts}, never requires to employ a catalyst which is not full rank.}

To see this, consider a bi-partite system with non-interacting Hamiltonian $H_1+H_2$.
Then consider the initial state $\rho_{SB}\otimes \sigma_C$ and apply an energy-preserving unitary operation $U$ that results in the state $\rho'_{SBC}$ with $\rho'_C=\sigma_C$.
Here, we imagine that $\rho_{SB}$ also includes the state of the heat-bath and thus there can be a built-up of correlations between the catalyst and $\rho_{SB}$.
Furthermore assume that $\rho_{SB}$ has full-rank and that $\sigma_C$ is supported only on a subspace $P\subset \mc H_2$ with complement $Q=\one - P$ (we identify the vector-space and the projector on the space). Thus $P=\sum_{j} \proj{j}$, where the sum is over the eigen-states of $\sigma_C$.
Let the spectrum of $\rho_{SB}$ and $\sigma_C$ be $\{p_\alpha\}$ and $\{q_j\}$, respectively. Then the final state
\begin{align}
\rho'_{SBC} = \sum_{\alpha,j} p_\alpha q_j U\proj{\alpha}\otimes \proj{j}U^\dagger
\end{align}
is a convex sum of the positive semi-definite operators $U\proj{\alpha}\otimes \proj{j}U^\dagger$. The sum has support only within $\one\otimes P$ since otherwise the reduced state $\rho'_C$ would also be supported outside of $P$. Hence also every summand is supported within $\one\otimes P$. Using $\one\otimes P = \sum_{\alpha,j} \proj{\alpha}\otimes\proj{j}$ we then obtain
\begin{align}
(\one\otimes Q) U(\one\otimes P) U^\dagger = \one \otimes Q \sum_{\alpha,j}U\proj{\alpha}\otimes\proj{j}U^\dagger = 0.
\end{align}
In other words, we have $(\one \otimes Q) U (\one \otimes P) = 0$ and by a similar calculation also   $(\one \otimes P) U(\one \otimes Q)=0$.
Thus, the unitary $U$ is block-diagonal. In particular, the operator $V=(\one \otimes P)U(\one \otimes P)$ considered as an operator on the Hilbert-space $\mc H_1\otimes P$ is unitary.
Since $U$ is energy-preserving by assumption, we can deduce that $P=\mathrm{span}\{\ket{E_j}\}$ for some subset of energy-eigenstates $\ket{E_j}$ of the Hamiltonian $H_2$ of the catalyst.

Then $V$ commutes with the Hamiltonian $H_1+ H_2|_{P}$, where $H_2|_{P}$ denotes the Hamiltonian of the catalyst, but restricted to the subspace $P$.

We can thus obtain an equivalent catalyst with full rank and a corresponding thermal operation by restricting $\sigma_C$ and $H_2$ to the subspace $P$ on which $\sigma_C$ has full rank and using the thermal operation defined by $V$:
\begin{align}
V \rho_{SB}\otimes \sigma_C|_P V^\dagger = \rho'_{SBC}|_{\one\otimes P}.
\end{align}
In particular note that the above analysis also shows that, in the case of exact catalysis, pure catalysts are useless: If a transition can be done with a pure catalyst, it can also be done without a catalyst.
}

\section{Approximate catalysis}
\label{app:approx_catalysts}
In this manuscript, we have assumed that catalysts are returned arbitrarily close to their initial state (or exactly, in the last section).
Here we will discuss possible relaxations of this assumption to include approximate catalysts.

First, we note that \ro{the problem of allowing for finite errors --in some suitable measure-- between the initial and final state of the catalyst is a delicate one, specially in the context of the third law of thermodynamics.}
The challenge is caused by the fact that already the statement of the unattainability principle is not stable under arbitrarily good approximations:
It compares the case where the state of the target system is exactly the ground-state with the case of approximating the  ground-state to arbitrary precision.
In the latter case, infinite resources are needed while in the latter case finite resources are needed (however diverging with the approximation precision).
This is the ultimate reason why a discontinuous measure of resources (like the vacancy) is necessary to capture the third law in the resource theoretic setting.

With this in mind, let us discuss the problem of approximate catalysts. If one demands that the catalyst is returned in approximately the same state,  it is crucial how one measures ``approximately''. In the context of thermal operations, this problem has been studied in Refs.~\cite{Brandao2015,Ng2015}. It has been shown in Ref.~\cite{Brandao2015}, that if one requires only that the catalyst is returned up to an arbitrarily small but fixed error in trace-distance, any transition can be implemented using a thermal operation to arbitrary precision -- without any resource. \ro{In particular, this implies that perfect cooling can be achieved without using any resource state.} Therefore, it is clear that stronger conditions are necessary to not trivialize \ro{problem of cooling.}

A second way to define approximate catalysts is to require that the catalyst is returned up to an error $\epsilon/\log(d)$ in trace-distance, where $d$ is the dimension of the catalyst, and $\epsilon>0$ is arbitrarily small but fixed for all catalysts.
Intuitively, this definition requires that the error is small even when multiplied by the number of particles in the catalyst.
In this case, transitions can be implemented to arbitrary precision if the non-equilibrium free energy decreases \cite{Brandao2015}.
This would lead to a constant amount of resources needed to cool a given system to arbitrary low temperatures -- hence the unattainability principle is \ro{also in this case} violated.

\ro{Because of the arguments above,} it seems \ro{allowing for a finite error --measured in trace-distance--} in catalyst seems to be too forgiving.
However, Ref.~\cite{Brandao2015} also hints at a solution to this problem: One should measure the error in terms of a quantity that is meaningful for the problem at hand. In Ref.~\cite{Brandao2015} the authors consider the problem of work extraction and demand in turn that the catalyst is returned with approximately the same ``work distance'', where the work distance measures the potential of one state to produce work. In our case, we are concerned with the task of low-temperature cooling. Indeed, the vacancy itself plays the role of the cooling potential, since the limitations for low temperature cooling of Thm. \ref{thm:necessary} are expressed in terms of the vacancy. We can thus require that the catalyst has to be returned with a vacancy that differs only by an amount $\epsilon$ from the initial vacancy. If we adopt this definition of approximate catalysts, the general necessary condition~\eqref{eq:thirdlawnecessary} is modified only slightly.
\he{This can be seen in the following way. First note that this notion requires that catalysts all have finite vacancy, i.e., they must have full rank. In this case, we can simply evaluate the vacancy of the resource, system and target before and after the cooling protocol has been applied. Let us assume that the initial state of the catalyst is $\sigma$, while the final state is $\sigma'$. Since the vacancy is an additive monotone of thermal operations and vanishes on thermal states, we then obtain
\begin{align}
\mc V_\beta(\rho_R\otimes \omega_\beta(H_S)\otimes \sigma,H_R+H_S+H_C) &= \mc V_\beta(\rho_R,H_R) + \mc V_\beta(\sigma,H_C) \\
&\geq \mc V_\beta(\omega_\beta(H_R)\otimes \rho_S \otimes \sigma', H_R+ H_S +H_C) = V_\beta(\rho_S,H_S) + \mc V_\beta(\sigma',H_C).
\end{align}
We hence obtain as a new necessary condition
\begin{align}\label{eq:necessary_approx}
\mc V_\beta(\rho_R,H_R) + \epsilon \geq V_\beta(\rho_S,H_S),
\end{align}
with $\epsilon = \mc V_\beta(\sigma,H_C) - \mc V_\beta(\sigma',H_C)$ being the error in the catalyst measured by the vacancy.
}

Thus the necessary condition and hence the quantitative unattainability principle is stable under approximate catalysts if defined consistently: \ro{allowing a fixed but small error measured by the vacancy difference.}

It seems plausible that under this definition of catalysts, also the sufficient condition in Theorem~\ref{thm:generaltheorem} simplifies to \eqref{eq:necessary_approx} for arbitrary resources -- at least for low enough target temperatures. However, proving this statement rigorously seems to require further technical innovations beyond the scope of this work. We therefore leave this as an open problem.

\end{document}